\documentclass[journal]{IEEEtran}
\usepackage{latexsym}
\usepackage{graphicx}
\usepackage{float}
\usepackage{graphicx,psfrag,amsmath,amsfonts,verbatim}
\usepackage{amssymb}
\usepackage[small]{caption}
\usepackage{indentfirst}
\usepackage{bm,bbding}
\usepackage[ruled,vlined,linesnumbered]{algorithm2e}
\usepackage{cite}
\usepackage[T1]{fontenc}
\usepackage[utf8]{inputenc}
\usepackage{authblk}
\usepackage{setspace}
\usepackage{fancyhdr}
\usepackage{multirow}
\usepackage{textcomp}
\newtheorem{theorem}{Theorem}
\newtheorem{lemma}{Lemma}

\def\BibTeX{{\rm B\kern-.05em{\sc i\kern-.025em b}\kern-.08em
    T\kern-.1667em\lower.7ex\hbox{E}\kern-.125emX}}
\begin{document}
\bibliographystyle{IEEEtran}
\allowdisplaybreaks[4]
\title{Resource Allocation for STAR-RIS-enhanced Metaverse Systems with Augmented Reality}
\author{Sun~Mao,
	Lei Liu,
	Kun Yang, 
	F. Richard Yu,
	Duist Niyato,
	and Chau Yuen
%Lei Liu, %\IEEEmembership{Member,~IEEE},
%Jie Hu, \IEEEmembership{Senior Member,~IEEE},
%Kun Yang, \IEEEmembership{Fellow,~IEEE},
%F. Richard Yu, \IEEEmembership{Fellow,~IEEE},
%and Duist Niyato, \IEEEmembership{Fellow,~IEEE}

\IEEEcompsocitemizethanks{This paper was funded in part by Natural Science Foundation of Sichuan Province under Grants 2025ZNSFSC0499 and 2025ZNSFSC0501, in part by Jiangsu Major Project on Fundamental Research under Grant BK20243059, in part by Gusu Innovation Project under Grant ZXL2024360, in part by High-Tech District of Suzhou City under Grant RC2025001, in part by Natural Science Foundation of China under Grants 62241108, 62132004, and 62531008, in part by MOE (Ministry of Education, Singapore), under MOE Tier 2 Award number T2EP50124-0032.\\
\mbox{~~~}Sun Mao is with Sichuan Internet College of Sichuan Normal University, Chengdu, China (e-mail: sunmao@sicnu.edu.cn).\\
\mbox{~~~}Lei Liu is with the State Key Laboratory of Integrated Service Networks, Xidian University, Xi’an 710071, China. (e-mail: tianjiaoliulei@163.com).\\
\mbox{~~~}Kun Yang  is with the State Key Laboratory of Novel Software Technology, Nanjing University, Nanjing 210008, China, and NJU Institute of Intelligent Networks and Communications (NINE), School of Intelligent Software and Engineering, Nanjing University (Suzhou Campus), Suzhou, 215163, China (e-mail: kunyang@nju.edu.cn).\\
\mbox{~~~}F. Richard Yu is with the School of Information Technology, Carleton University, Ottawa, ON K1S 5B6, Canada. (e-mail: richard.yu@ieee.org).\\
\mbox{~~~}Duist Niyato is with School of Computer Science and Engineering, Nanyang Technological University, 639798, Singapore (e-mail: dniyato@ntu.edu.sg).\\
\mbox{~~~}Chau Yuen is with the School of Electrical and Electronic Engineering, Nanyang Technological University, 639798, Singapore (e-mail: chau.yuen@ntu.edu.sg).
 }}
\maketitle

\begin{abstract}
Augmented reality (AR)-enabled Metaverse is a promising technique to provide immersive service experience for mobile users. However, the limited network resources and unpredictable wireless propagation environments are key design bottlenecks of AR-enabled Metaverse systems. Therefore, this paper presents a resource management framework for simultaneously transmitting and reflecting RIS (STAR-RIS)-assisted AR-enabled Metaverse, where the STAR-RIS is configured to improve the communication efficiency between AR users and the Metaverse server located at the base station (BS). Moreover, we formulate a service latency minimization problem via jointly optimizing the computation resource allocation of the BS, coefficient matrix of the STAR-RIS, central processing unit (CPU) frequency and transmit power of the AR users. To tackle the non-convex problem, we utilize an approximate method to transform it to a tractable form, and decouple the multi-dimensional variables via the alternating optimization method. Particularly, the optimal coefficient matrix is obtained by a penalty function-based method with proved convergence, the CPU frequencies of AR users are derived as the closed-form solution, and the transmit power of AR users and computation resource allocation of the BS are obtained by the Lagrange duality method and convex optimization theory. Finally, simulation results demonstrates that the proposed method achieves remarkable latency reduction than several benchmark methods.
\end{abstract}
\IEEEpeerreviewmaketitle
\section{Introduction}
In recent years, Metaverse is envisioned as an innovative technique to transform the way people live and work, which can be applied in multiple domains, such as intelligent transportation, interactive game, industry automation, online education, smart cities, etc \cite{10070414}. As an important technique for Metaverse, augmented reality (AR) technique is able to provide immersive service experience to mobile users by blending physical and virtual reality \cite{de2020survey}.

In AR-enabled Metaverse applications, object detection is an important function to improve the quality of experience of mobile users \cite{liu2019edge}. Imagine a vehicular network that the real-time object detection on the road can improve the driving safety of vehicles, especially in the extreme weather and lighting conditions. In general, the object detection is achieved by utilizing intelligent machine learning frameworks with heavy computation requirements. However, it is hard to satisfy the strict computation requirements using the local computing capability of resource-constrained mobile users. Therefore, mobile-edge computing (MEC) is envisioned as a promising technique to address this problem, where the AR users can transmit the collected image/video to the BS integrated with edge Metaverse servers, which can perform the object detection and return the corresponding results to AR users \cite{8016573}.

Nevertheless, the communication efficiency between AR users and BS cannot be always ensured due to stochastic wireless channels and user mobility. In recent years, reconfigurable intelligent surface (RIS) is proposed to reconstruct the wireless propagation environments via establishing the reliable cascade links between the BS and AR users \cite{9136592,9410457,9500188}. Although the benefits brought by conventional reflecting-only RIS, it only can achieve 180-degree coverage, and the communication performance of AR users located at the opposite side of RIS cannot be enhanced. With the advancement of information metamaterial, the simultaneous transmitting and reflecting RIS (STAR-RIS) is envisioned as a promising solution to address this problem, since it can reflect and transmit the incident signals simultaneously via configuring appropriate transmit/reflect coefficient matrix \cite{9690478,9754364}.

Nowadays, the resource allocation for AR-empowered Metaverse systems has been investigated in \cite{7906521,9319727,feng2022resource,zhou2022resource}. Moreover, the authors in \cite{9631876,9197675,9786814} and \cite{9570143,9740451,9854887,10093070} utilized reflecting-only RIS and STAR-RIS to improve the communication performance (such as total throughput, energy efficiency, etc.), respectively. Meanwhile, existing literature in \cite{9388935,10306290,10032506,10121446} proposed joint computation offloading and RIS configuration optimization methods for RIS-aided MEC networks. For ease of reference, the related literature is summarized in Table I. As observed, although the research efforts brought by existing work, several problems are still unsolved as follows.
\begin{table*}[htbp]
\footnotesize
\caption{Synopsis of relevant research works}
\centering
\begin{tabular}{|lllllll|}
\hline Ref. &Optimization target & Metaverse &MEC &RIS &Communication optimization &Computation optimization \\
\hline \cite{7906521,9319727} &Energy consumption &\Checkmark &\Checkmark &\XSolidBrush &\Checkmark &\Checkmark\\
\hline \cite{feng2022resource} &System utility &\Checkmark &\Checkmark &\XSolidBrush &\Checkmark &\Checkmark\\
\hline \cite{zhou2022resource} &Service latency &\Checkmark &\Checkmark &\XSolidBrush &\Checkmark &\Checkmark\\
\hline \cite{9631876,9197675,9786814} & Communication performance &\XSolidBrush &\XSolidBrush &Reflecting-only RIS &\Checkmark &\XSolidBrush \\
\hline \cite{9570143,9740451,9854887,10093070} & Communication performance &\XSolidBrush &\XSolidBrush &STAR-RIS &\Checkmark &\XSolidBrush \\
\hline \cite{9388935} & Energy consumption &\XSolidBrush &\Checkmark &Reflecting-only RIS &\Checkmark &\Checkmark \\
\hline \cite{10306290} & Total computation bits &\XSolidBrush &\Checkmark &Reflecting-only RIS &\Checkmark &\Checkmark \\
\hline \cite{10032506,10121446} & Total computation bits  &\XSolidBrush &\Checkmark &STAR-RIS &\Checkmark &\Checkmark \\
\hline This article &Min-max service latency &\Checkmark &\Checkmark &STAR-RIS &\Checkmark &\Checkmark\\
\hline
\end{tabular}
\end{table*}
\begin{itemize}
\item To the best of our knowledge, the integration of STAR-RIS into AR-empowered Metaverse systems has not been investigated in existing literature. In AR-empowered Metaverse systems, AR users need to communicate frequently with Metaverse servers located at the BS via stochastic wireless channels, for enriching the virtual scenarios and improving the service experience of  Metaverse systems. Benefited from the ability of reconstructing wireless propagation environments, STAR-RIS exhibits great potential to improve the communication performance between AR users and the BS. Different from existing literature focusing on system communication performance or computation capacity, the user-centric performance requirements and random user distribution incur new and great design challenges on the optimal configuration of STAR-RIS in AR-empowered Metaverse systems, especially considering the tightly coupled relationship between transmit and reflect coefficient matrix of STAR-RIS.
\item In AR-empowered Metaverse systems, real-time object detection is regarded as an important function to improve the quality of experience (QoE) of AR users. In general, the service latency of object detection is composed by local processing delay, data transmission delay, edge execution and result delivering delay, which is recognized as an important performance metric to evaluate the QoE of AR users. Obviously, the service latency should be controlled rigidly to avoid discomfort of AR users (e.g., motion sickness), and to reduce the possibility of terrible accidents in delay-sensitive scenarios, such as Internet of Vehicles (IoVs). Therefore, it is essential to model and optimize the object detection latency in AR-empowered Metaverse systems, which is seldom investigated in existing literature.
\item In STAR-RIS-enhanced AR-empowered Metaverse systems, the object detection is involved with the process of communication and computing. Therefore, in order to meet the strict delay requirements, it is important to jointly optimize multi-dimensional network resources and reflect/transmit coefficient matrix of STAR-RIS. Nevertheless, due to highly coupled optimization variables, it is extremely difficult to design joint reflect/transmit beamforming and resource management algorithms.
\end{itemize}

Motivated by these observations, this paper focuses on optimizing the coefficient matrices of STAR-RIS and multi-dimensional network resources to reduce the service latency of AR-empowered Metaverse systems. In this article, we consider a typical STAR-RIS-assisted AR-enhanced Metaverse system, which includes a BS integrated with edge Metaverse server, a STAR-RIS, and multiple AR users. The main contribution of this article is summarized as follows.
\begin{itemize}
\item This paper presents a novel design framework of STAR-RIS-assisted AR-enabled Metaverse, where AR users transmit the collected environmental data to the edge Metaverse server located at BS for enriching the virtual scenarios, and meanwhile AR users can enjoy diverse applications provided by Metaverse. In this framework, the Metaverse concentrates on an application of object detection, where the AR users first transmit image/video to the BS, and then the BS performs object detection and return the result to AR users. Moreover, both the spatial division multiple access (SDMA) and frequency division multiple access (FDMA) are exploited as multiple-access protocols. Meanwhile, the transmit/reflect coefficient matrix of STAR-RIS can be optimally configured to enhance the communication performance between the BS and AR users.
\item To reduce the object detection latency, we first model the service latency including image conversion time at AR users, data transmission delay, and edge execution and result downloading delay. Then, we formulate a service latency minimization problem for both SDMA-based and FDMA-based Metaverse systems, while satisfying wireless resource restrictions and user energy constraints, through jointly optimizing local CPU-cycle frequency, bandwidth allocation ratio and transmit power of AR users, transmit/reflect coefficient matrix of the STAR-RIS, and computation resource assignment of the BS.
\item With the tightly coupled optimization variables, the formulated problems are strictly non-convex. To tackle it, we first present an approximate method to rewrite the non-convex problem to a tractable form. Then, we design an alternating optimization method to decouple optimization variables. In particular, the penalty function-based method is developed to obtain the optimal transmit/reflect coefficient matrix of STAR-RIS. Moreover, the local computing frequencies of AR users are derived as the closed-form solution. The Lagrange duality method and convex optimization technique is adopted to obtain the optimal computation resource allocation of BS, bandwidth allocation ratio and transmit power of AR users. Finally, convergence and computational complexity of proposed algorithms are analyzed.
\item Theoretic analysis reveals that each AR user enjoys the same service latency to minimize the maximum service latency of AR users; the edge Metaverse server located at the BS adopts the maximum computing capability to reduce service latency of AR users. Meanwhile, extensive numerical results show that the deployment and configuration of STAR-RIS has great potential to reduce the service latency of AR users, especially in the scenario with restricted communication resources and heavy transmission requirements.
\end{itemize}

The rest of this paper is organized as follows. Section II illustrates the system model and the formulated service latency minimization problem. Section III proposes the joint coefficient configuration and resource scheduling algorithm for solving the formulated non-convex optimization problem. Section IV extends this work to FDMA-based Metaverse systems. Section V presents simulation results to demonstrate the superiority of proposed method. This article is concluded in Section VI.
\section{System Model}
As shown in Fig. 1, this paper considers a typical STAR-RIS-assisted AR-enabled Metaverse system including a BS, a STAR-RIS with $N$ elements, and $K$ AR users \footnote{For ease of analysis, this paper investigates the optimal resource management and transmit/reflect configuration optimization for a single STAR-RIS-aided AR-empowered Metaverse system, and the proposed method will provide meaningful guidance for the optimal design of the scenario with multiple STAR-RISs.}. The sets of transmit/reflect elements and AR users are denoted as $\mathcal{N}=\{1,2,\cdots,N\}$ and $\mathcal{K}=\{1,2,\cdots,K\}$, respectively. The Metaverse service is provided by the Metaverse server located at the BS. In addition, it is assumed that the AR-enabled Metaverse concentrates on AR applications for object detection. The AR user first transmits the image/video obtained by on-board camera to the Metaverse system. After receiving the image/video from AR users, the object detection is conducted at the Metaverse server, and then the result is returned to the AR users and rendered on the AR screen. In particular, the STAR-RIS is equipped to improve the transmission rate between the BS and AR users, through utilizing its advantage to reconstruct wireless propagation environments.
%It should be noted that AR users transmit their collected environment information to Metaverse servers for enriching the scene of Metaverse, and they also fully utilize the object detection service provided by Metaverse systems for obtaining a more immersive user experience.
\begin{figure}[htbp]
\begin{center}
\includegraphics[width=0.5\textwidth]{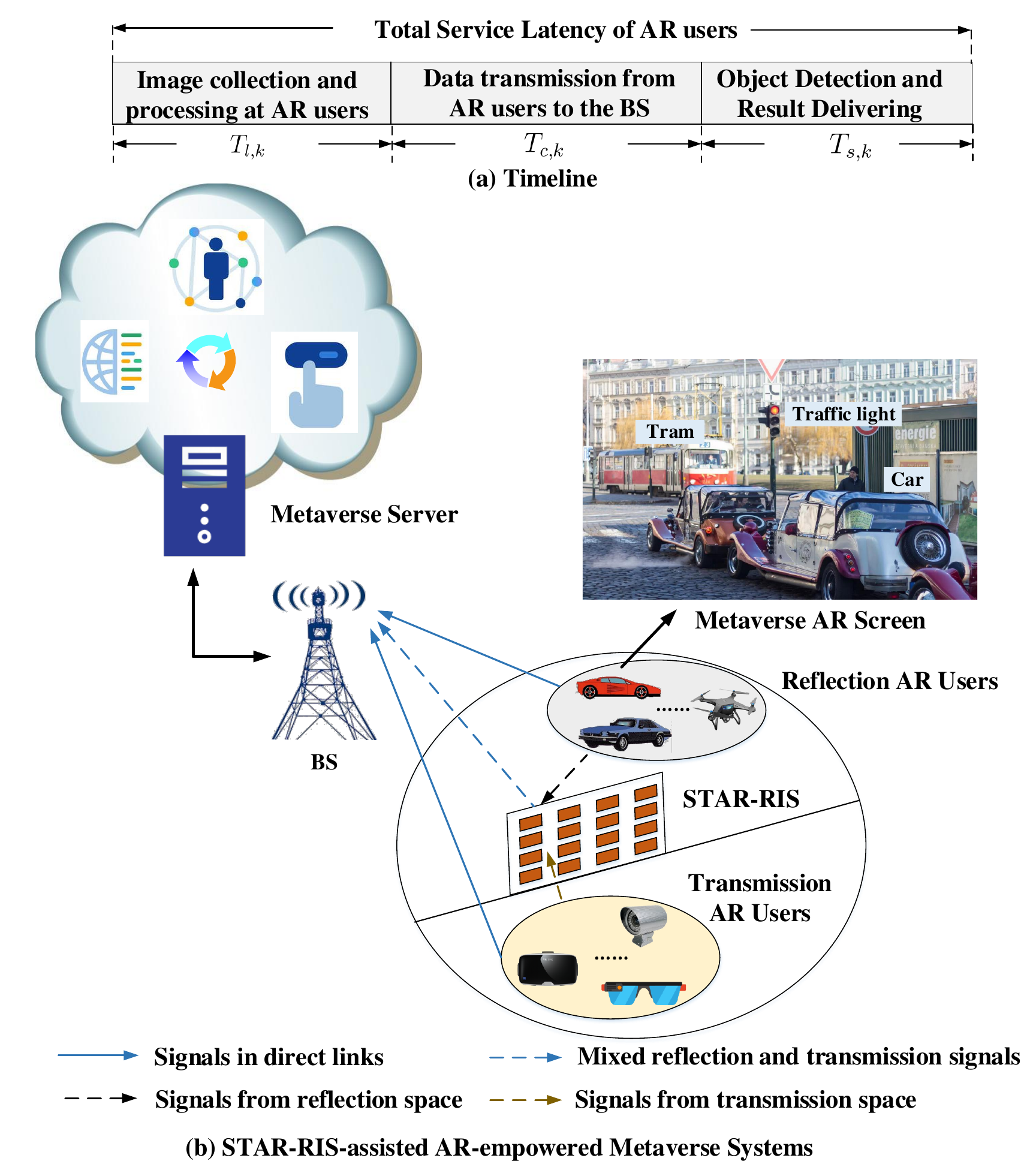}
\caption{System model.}
\end{center}
\end{figure}

Notice that the Metaverse systems generally include the mobile users requiring Metaverse services in the downlink. Nevertheless, this paper concentrates on how the Metaverse servers and AR users together build Metaverse services, i.e., AR users submit captured environmental information to Metaverse servers for enriching the scene of Metaverse, while they enjoy the object detection services provided by Metaverse systems for obtaining an intelligent and immersive experience. Involving the requesting users in the downlink may be an interesting future direction. Moreover, software defined network (SDN) technique can be adopted to achieve efficient network management \cite{9583902,9491087}. The workflow of SDN-based Metaverse systems is briefly summarized as follows.
\begin{itemize}
\item At the beginning of each time slot, the BS obtains the service requests of AR users and perfect channel state information of all links \footnote{In general, it is difficult to perform channel estimation for STAR-RIS-aided wireless communication systems considering the passive feature of STAR-RIS. In order to fulfill the potential of STAR-RIS, similar to \cite{10032506,10121446}, it is assumed that the perfect CSI can be obtained by utilizing existing channel estimation methods \cite{9130088,9505267,9087848,9839429}.} .
\item After receiving the global network information, the SDN controller at the BS updates the local information table, and designs the corresponding STAR-RIS configuration optimization and resource management strategy to minimize the service latency of AR users.
\item According to the scheduling strategy, the ARs users, STAR-RIS and BS adjust the relevant transmit power, CPU-cycle frequency, transmit/reflect coefficient matrix, and computation resource allocation.
\end{itemize}

It is worth noting that the considered scenario can be found in practical systems. For instance, the proposed framework can be utilized in Internet of Vehicles for improving the driving experience, where AR-empowered vehicles can enjoy the object detection services provided by Metaverse systems, and meanwhile it can also enrich the virtual scenarios by transmitting the collected road information to Metaverse servers. For ease of reference, the main symbols used in this article are summarized in Table II.

\begin{table}[htbp]
\footnotesize
\caption{Summary of major notations}
\centering
\begin{tabular}{|l|l|}
\hline Notation & Description\\
\hline $K$ & Number of AR users\\
\hline $K_t$/$K_r$ & Number of reflection/transmission AR users\\
\hline $N$ & Number of STAR-RIS elements\\
\hline $f_{k,\text{max}}$ & Maximum CPU-cycle frequency of $k$-th AR user \\
\hline $w_k$ & Local computation workload of $k$-th AR user\\
\hline $f_{l.k}$ & Local CPU-cycle frequency of $k$-th AR user\\
\hline $T_{l,k}$ & Local computing latency of $k$-th AR user\\
\hline $E_{l,k}$ & Energy consumption for local computing of $k$-th AR user\\
\hline $p_{k,\text{max}}$ & Maximum transmission power of $k$-th AR user\\
\hline $p_k$ & Transmission power of $k$-th AR user\\
\hline $B$ & Total transmission bandwidth\\
\hline $R_k$/$R_k^{\text{FDMA}}$  & Transmission rate of $k$-th AR user\\
\hline $\delta^2$ & Gaussian noise power spectral density \\
\hline $h_{d,k}$ & Channel coefficient between BS and $k$-th AR user\\
\hline $\mathbf{h}_r$ & Channel coefficient between BS and STAR-RIS\\
\hline $\mathbf{h}_{I,k}$ & Channel coefficient between BS and $k$-th AR user\\
\hline $\Gamma_r$/~$\Gamma_t$ & Reflection/Transmission coefficient matrix of STAR-RIS\\
\hline $\gamma_{r,n}$/~$\gamma_{t,n}$ & Amplitude of $n$-th element on STAR-RIS for\\
& reflection/transmission\\
\hline $\theta_{r,n}$/~$\theta_{r,n}$ & Phase of $n$-th element on STAR-RIS for \\
& reflection/transmission\\
\hline $T_{c,k}$ & Transmission delay of $k$-th AR user \\
\hline $d_k$ & Computation model size\\
\hline $E_{c,k}$ & Communication energy consumption of $k$-th AR user\\
\hline $E_k$ & Total energy consumption of $k$-th AR user\\
\hline $c_k$ & Number of CPU cycles required for processing 1-bit data\\
\hline $T_{s,k}$ & Edge processing latency for analyzing data from $k$-th \\
&AR user\\
\hline $f_{s,k}$ & Computation resource of BS allocated to $k$-th AR user\\
\hline $F$ & Maximum CPU-cycle frequency of BS\\
\hline $b_k$ & Bandwidth allocation ratio of $k$-th AR user \\
\hline
\end{tabular}
\end{table}
\subsection{Local Processing Model}
Since the images obtained by on-board cameras are generally YUV format, so the AR users need to execute the pre-processing task for converting YUV to RGB \cite{yang2007fast}. In general, the CPU of AR users is utilized to perform the format conversion, and the dynamic voltage and frequency scaling (DVFS) is adopted to reduce power consumption by adjusting clock frequency and supply voltage \cite{eyerman2011fine}. Defining $w_k$ as the computation workload of $k$-th AR user, the processing latency at $k$-th AR user is given by $T_{l,k}=\frac{w_k}{f_{l,k}}$, where $f_{l,k}$ represents the CPU-cycle frequency of $k$-th AR user. Moreover, the corresponding computation energy consumption of $k$-th AR user can be expressed as $E_{l,k}=\kappa f_{l,k}^3T_{l,k}=\kappa f_{l,k}^2w_k$, where $\kappa$ is a constant related to the circuit architecture.
\subsection{STAR-RIS-aided Communication Model}
After accomplishing local format conversion, AR users need to transmit the obtained image/video to Metaverse server with the aid of STAR-RIS. In the image transmission stage, the energy splitting protocol is utilized in STAR-RIS, and the spatial division multiple access (SDMA) is adopted to achieve multi-user communications. As illustrated in Fig. 1, according to the location of STAR-RIS, the AR users located at the same side of the BS are called reflection AR users, and the other users are called transmission AR users. Defining $K_t$ and $K_r$ as the number of transmission AR users and reflection AR users, respectively. The coefficient matrices of STAR-RIS for reflection and transmission modes are denoted by $\Gamma_r=\text{diag}\{\sqrt{\gamma_{r,1}} e^{j\theta_{r,1}},\sqrt{\gamma_{r,2}}e^{j\theta_{r,2}},\cdots,\sqrt{\gamma_{r,N}}e^{j\theta_{r,N}}\}$ and $\Gamma_t=\text{diag}\{\sqrt{\gamma_{t,1}}e^{j\theta_{t,1}},\sqrt{\gamma_{t,2}}e^{j\theta_{t,2}},\cdots,\sqrt{\gamma_{t,N}}e^{j\theta_{t,N}}\}$, respectively, where $\gamma_{r,n}$ and $\gamma_{t,n}$ stand for the amplitude of $n$-th element for reflection and transmission, respectively, $\theta_{r,n}$ and $\theta_{t,n}$ represent the phase of $n$-th element for reflection and transmission, respectively. For the conventional passive STAR-RIS, the amplitude should satisfy $\gamma_{t,n}+\gamma_{r,n}=1, \forall n\in\mathcal{N}$. In addition, if we set $\gamma_{t,n}=0$, the STAR-RIS will be simplified to a conventional reflecting-only RIS. Therefore, the latter is envisioned as a special case of STAR-RIS. Denoting $p_k$ as the transmission power of $k$-th AR user, the received signal at the BS is expressed as
\begin{equation}
\begin{aligned}
y=\sum\limits_{k\in\mathcal{K}}(h_{d,k}+\mathbf{h}_r^H\Gamma_{i(k)} \mathbf{h}_{I,k})\sqrt{p_k}s_k+n_0,
\end{aligned}
\end{equation}
where $i(k)=r$ and $i(k)=t$ indicate that the $k$-th AR user is located at the reflection space and transmission space, respectively, $s_k$ denotes the information symbol of $k$-th AR user that satisfies $\mathbb{E}(s_k)=0$, $\mathbb{E}(s_ks_k^H)=1$ and $\mathbb{E}(s_ks_m^H)=0$, $\forall k\neq m$, $n_0$ is the Gaussian noise at BS with power spectral density $\delta^2$, $h_{d,k}=\sqrt{\rho d^{-\iota_{d,k}}_{d,k}}(\sqrt{\frac{\varsigma}{\varsigma+1}}+\sqrt{\frac{1}{\varsigma+1}}\hat{h_{d,k}})$, $\mathbf{h}_{I,k}=\sqrt{\rho d^{-\iota_{l,k}}_{I,k}}(\sqrt{\frac{\varsigma}{\varsigma+1}}+\sqrt{\frac{1}{\varsigma+1}}\hat{\mathbf{h}_{I,k}})$ and $\mathbf{h}_r=\sqrt{\rho d^{-\iota_r}_{r}}(\sqrt{\frac{\varsigma}{\varsigma+1}}+\sqrt{\frac{1}{\varsigma+1}}\hat{\mathbf{h}_{r}})$ indicate the channels from the $k$-th AR user to the BS, from the $k$-th AR user to the STAR-RIS, and from the STAR-RIS to the BS, respectively, where $d_{d,k}$, $d_{I,k}$ and $d_{r}$ represent the distance between the $k$-th AR user and the BS, between the $k$-th AR user and the STAR-RIS, and between the STAR-RIS and the BS, respectively, $\rho$ denotes the path-loss at the unit distance,  $\iota_{d,k}$/$\iota_{l,k}$/$\iota_{r}$ indicates the path-loss factor, $\varsigma$ is the Rician factor, and $\hat{h_{d,k}}$, $\hat{\mathbf{h}_{I,k}}$ and $\hat{\mathbf{h}_{r}}$ follow complex Gaussian distribution with zero mean and unit variance \footnote{It is assumed that the BS can acquire perfect channel state information (CSI) using existing channel estimation techniques \cite{9722893}, thereby enabling the evaluation of the achievable minimum service latency in the considered STAR-RIS-enhanced Metaverse systems. The results can also serve as a useful reference for designing robust resource allocation and STAR-RIS configuration strategies in scenarios with imperfect CSI.}.
Defining $B$ as the system bandwidth, the transmission rate of $k$-th AR user is given by
\begin{equation}
\begin{aligned}
&R_{k}=B\log_2\left(1+\frac{p_k|h_{d,k}+\mathbf{h}_r^H\Gamma_{i(k)} \mathbf{h}_{I,k}|^2}{\sum\limits_{j\neq k}p_j|h_{d,j}+\mathbf{h}_r^H\Gamma_{i(j)} \mathbf{h}_{I,j}|^2+B\delta^2}\right), \\
&\forall k\in\mathcal{K}.
\end{aligned}
\end{equation}
Denoting $d_k$ (pixels) as the computation model size of $k$-th AR user, so the BS will receive the image frame with resolution $d_k\times d_k$ for object detection. Then, the communication delay for $k$-th AR user is expressed as $T_{c,k}=\frac{\beta d_k^2}{R_{k}}$, where $\beta$ represents the bits of information carried by a pixel. In addition, the corresponding communication energy consumption is $E_{c,k}=p_kT_{c,k}$. The total energy consumption of $k$-th AR user mainly includes communication and computation energy consumption, which is denoted as
\begin{equation}
E_k=E_{l,k}+E_{c,k},\forall k\in\mathcal{K}.
\end{equation}
\subsection{Edge Processing at Metaverse Server}
After receiving the image from AR users, the BS needs to conduct the object detection. Defining $f_{s,k}$ as the edge computation capability allocated to the $k$-th AR user, so the edge processing latency to analyze the image of $k$-th AR user is expressed as $T_{s,k}=\frac{\beta d_k^2c_k}{f_{s,k}}$, where $c_k$ stands for the number of CPU cycles required to process 1-bit data\footnote{It is worth emphasizing that this work concentrates on the design of an efficient multi-dimensional resource management framework aimed at minimizing user service latency in STAR-RIS-assisted AR-empowered Metaverse systems. Rather than tailoring the study to specific object detection architectures, we follow the modeling approach in \cite{feng2022resource,10959034} to characterize the computational demand of object detection. Based on this, we formulate the service latency minimization problem and develop the corresponding resource management algorithm, which are presented in the sequel. Furthermore, the proposed framework can be seamlessly integrated with a wide range of detection models, while ensuring reduced latency through optimized resource allocation.} Hence, the total service latency of $k$-th AR user is given by
\begin{equation}
T_k=T_{l,k}+T_{c,k}+T_{s,k}, \forall k\in\mathcal{K}.
\end{equation}
It should be noted that the result downloading is ignored considering the small size of results and powerful communication capability of BS \cite{8937028}.
%The energy consumption at the BS for processing the image from $k$-th AR user is expressed as $E_{s,k}=\kappa_s f_{s,k}^3T_{s,k}=\kappa_s \beta d_k^2c_kf_{s,k}^2$, where $\kappa_s$ is a constant associated with the hardware architecture.
\subsection{Problem Formulation}
This paper aims at minimizing the maximum service latency of AR users while restricting AR users' maximum energy consumption \footnote{To reveal the performance limit, the individual service latency constraint of each AR user is not integrated in formulated problem (5). Noted that the proposed method in this paper can be extended to the scenario considering the individual service latency constraint of each AR user.}, via jointly optimizing the CPU-cycle frequency  $\{f_{l,k}\}$ of AR users, the transmission power $\{p_k\}$ of AR users, the transmit/reflect coefficient matrix $\{\Gamma_t,\Gamma_r\}$ of STAR-RIS, and the computation resource assignment $\{f_{s,k}\}$ of BS. Mathematically, the corresponding optimization problem is expressed as
\begin{equation}
\begin{split}
\underset{\{f_{l,k}, p_{k}, \Gamma_t, \Gamma_r, f_{s,k}\}}{\text{minimize }}~~~& \underset{k\in\mathcal{K}}{\text{max}}~\{T_{l,k}+T_{c,k}+T_{s,k}\}\\
\text{s.t. }~~~~~~
&\text{C1: } 0\leq f_{l,k}\leq f_{k,\text{max}}, \forall k\in\mathcal{K},\\
&\text{C2: } 0\leq p_k\leq p_{k,\text{max}}, \forall k\in\mathcal{K},\\
&\text{C3: } 0\leq \theta_{i,n}\leq 2\pi, \forall i\in\{t,r\}, n\in\mathcal{N},\\
&\text{C4: } \gamma_{t,n}+\gamma_{r,n}=1, \forall n\in\mathcal{N},\\
&\text{C5: } E_{l,k}+E_{c,k}\leq E_{k,\text{max}},\forall k\in\mathcal{K},\\
&\text{C6: } \sum\limits_{k\in\mathcal{K}}f_{s,k}\leq F,\\
&\text{C7: } f_{s,k}\geq 0, \forall k\in\mathcal{K},\\
&\text{C8: }\gamma_{i,n}\geq 0, \forall i\in\{t,r\}, n\in\mathcal{N},\\
\end{split}
\end{equation}
where $f_{k,\text{max}}$, $p_{k,\text{max}}$ and $E_{k,\text{max}}$ denote the maximum CPU-cycle frequency, transmit power and energy consumption of AR users, respectively. In (5), C1 restricts the maximum CPU-cycle frequencies of AR users, C2 implies the transmission power constraints of AR users, C3 and C4 stand for the phases and amplitudes constraints of STAR-RIS, C5 represents the maximum energy consumption constraints of AR users, C6 and C7 indicate the computation resource constraint of BS, and C8 denotes the amplitudes of STAR-RIS should be non-negative.

\emph{Remark 1: } In considered AR-empowered Metaverse systems, service latency of AR users is an important performance metric to achieve real-time object detection, especially in the delay-sensitive scenario. For instance, high latency for object detection may cause serious traffic accidents in IoVs. Moreover, a lower service latency will greatly reduce the possibility of discomfort of AR users, e.g., motion sickness. Better service experience will attract more AR users to join the Metaverse systems and to enrich the virtual scenarios by sharing their collected environmental data. Therefore, this paper adopts the maximum service latency of AR users as our optimization objective.

\emph{Remark 2: } It is worth emphasizing that this paper concentrates on a representative scenario with static AR users to enable tractable analysis, thereby revealing the fundamental impact of the proposed joint resource management and STAR-RIS optimization design on the maximum service latency across all AR users in Metaverse systems. Nonetheless, the proposed framework also offers valuable insights for the long-term optimization of STAR-RIS-enhanced Metaverse systems with user mobility, since the long-term problem can be decomposed into a sequence of per–time-slot optimization subproblems.

\emph{Remark 3: } Unlike the total service latency minimization problem focusing on reducing the overall latency of all AR users, this paper aims at minimizing the maximum service latency experienced by any individual AR user. Our proposed method tends to allocate the network resources to ensure that no AR user experiences excessively high latency, even if it means that the overall latency of all AR users might be slightly higher. Considering the user-centric service requirements of Metaverse systems, the proposed method is suitable for conducting network resource management and STAR-RIS configuration optimization for AR-empowered Metaverse systems.

\emph{Remark 4: }The formulated maximum service latency minimization problem (5) is strictly non-convex due to the following reasons: 1) Because of the presence of harmful co-channel interference and the fractional structure, the service latency and energy consumption expressions are strictly non-convex and non-concave; 2) The highly coupled among optimization variables, such as $\Gamma_i$ and $p_k$, $\gamma_{t,n}$ and $\gamma_{r,n}$, etc. Therefore, conventional convex optimization methods, such as Lagrange duality method and interior-point method, cannot be directly used to obtain the optimal solution of (5). Furthermore, the heuristic algorithms are generally high-complexity and provide uncertain convergence rate.
\section{Proposed Solution}
To solve the formulated non-convex problem (5), we first utilize an approximate method to rewrite the non-convex and non-concave transmission rate expression to a  tractable concave form. Then, to tackle the coupled optimization variables, we develop an alternating optimization method to optimize transmit power of AR users, transmit/reflect coefficient matrix of the STAR-RIS, local CPU frequencies of AR users, and computation resource allocation of the BS. Particularly, the transmit/reflect coefficient matrix of STAT-RIS is obtained by designing a penalty function-based method with proved convergence, the CPU-cycle frequencies of AR users are derived in closed-form solution, Lagrange duality method and convex optimization theory are adopted to obtain the optimal transmit power of AR users and computation resource assignment of the BS. In addition, the convergence and computational complexity are further analyzed.
%The detailed flowchart for solving (5) is illustrated in the following Fig. 2.
%\begin{figure}[htbp]
%\begin{center}
%\includegraphics[width=0.5\textwidth]{Fig2.pdf}
%\caption{The flowchart for solving service latency minimization problem (5).}
%\end{center}
%\end{figure}

To remove the max function in (5), we define $t=\underset{k\in\mathcal{K}}{\text{max}}~\{T_{l,k}+T_{c,k}+T_{s,k}\}$ to reformulate (5) to the following equivalent problem
\begin{subequations}
\begin{align}
\underset{\{f_{l,k}, p_{k}, \Gamma_t, \Gamma_r, f_{s,k},t\}}{\text{minimize }}~~~& t\\
\text{s.t. }~~~~~~~&T_{l,k}+T_{c,k}+T_{s,k}\leq t, \forall k\in\mathcal{K},\\
&\text{C1-C8.}
\end{align}
\end{subequations}

As observed in (2), the uplink transmission rate $R_k$ is non-convex and non-concave function due to the coupled optimization variables and harmful co-channel interference. Based on the Fenchel conjugate arguments \cite{9295408}, the following lemma allows us to rewrite $R_k$ for obtaining optimally solvable sub-problems in next subsections.
%To tackle it, the following \emph{lemma 1} is introduced to reformulate $R_k$.

\begin{lemma}
Denoting $f(b)=-ba+\ln(b)+1$, it yields that
\begin{equation}
-\ln a= \underset{b>0}{\text{max }} f(b).
\end{equation}
\end{lemma}
The equality holds when $b=1/a$. Defining $a=\sum\limits_{j\neq k}p_j|h_{d,j}+\mathbf{h}_r^H\Gamma_{i(j)} \mathbf{h}_{I,j}|^2+B\delta^2$ and $b=y_k^{(n)}$, $R_{k}$ is transformed as
\begin{equation}
\begin{aligned}
&R_{k}= B\log_2\left(1+\frac{p_k|h_{d,k}+\mathbf{h}_r^H\Gamma_{i(k)} \mathbf{h}_{I,k}|^2}{\sum\limits_{j\neq k}p_j|h_{d,j}+\mathbf{h}_r^H\Gamma_{i(j)} \mathbf{h}_{I,j}|^2+B\delta^2}\right)\\
&=\frac{B}{\ln 2}\ln(\frac{\sum\limits_{k\in\mathcal{K}}p_k|h_{d,k}+\mathbf{h}_r^H\Gamma_{i(k)} \mathbf{h}_{I,k}|^2+B\delta^2}{\sum\limits_{j\neq k}p_j|h_{d,j}+\mathbf{h}_r^H\Gamma_{i(j)} \mathbf{h}_{I,j}|^2+B\delta^2})\\
&=\frac{B}{\ln 2}(\ln(\sum\limits_{k\in\mathcal{K}}p_k|h_{d,k}+\mathbf{h}_r^H\Gamma_{i(k)} \mathbf{h}_{I,k}|^2+B\delta^2)\\
&-\ln(\sum\limits_{j\neq k}p_j|h_{d,j}+\mathbf{h}_r^H\Gamma_{i(j)} \mathbf{h}_{I,j}|^2+B\delta^2))\\
&\overset{(\alpha_1)}{=}\frac{B}{\ln 2}(\ln(\sum\limits_{k\in\mathcal{K}}p_k|h_{d,k}\!+\!\mathbf{h}_r^H\Gamma_{i(k)} \mathbf{h}_{I,k}|^2+B\delta^2)\\
&-y_k^{(n)}(\sum\limits_{j\neq k}p_j|h_{d,j}\!+\mathbf{h}_r^H\Gamma_{i(j)} \mathbf{h}_{I,j}|^2\!+B\delta^2)\!+\ln(y_k^{(n)})\!+1),
\end{aligned}
\end{equation}
where $(\alpha_1)$ holds due to the transformation of $-\ln(\sum\limits_{j\neq k}p_j|h_{d,j}+\mathbf{h}_r^H\Gamma_{i(j)} \mathbf{h}_{I,j}|^2+B\delta^2)$ via \emph{Lemma 1}.

Then, we design an alternating optimization method to decouple the optimization variables into four blocks, namely $\{\{p_k\},\{\Gamma_t,\Gamma_r\},\{f_{l,k}\},\{f_{s,k}\}\}$. Then, the transmit power $\{p_k\}$ of AR users, transmit/reflect coefficient matrix $\{\Gamma_t,\Gamma_r\}$ of the STAR-RIS, CPU-cycle frequency $\{f_{l,k}\}$ of AR users, and computation resource allocation $\{f_{s,k}\}$ of the BS are alternately optimized, while keeping the other three blocks of variables fixed. After obtaining the optimal solution, $y_k^{(n)}$ is updated at the end of each iteration.
\subsection{Transmit Power Control}
Under given $\{f_{l,k}^{(n)}, \Gamma_t^{(n)}, \Gamma_r^{(n)}, f_{s,k}^{(n)}\}$ in the $n$-th iteration, (6) is converted to the power control subproblem
\begin{subequations}
\begin{align}
\underset{\{p_{k}, t\}}{\text{minimize }}~~~& t\\
\text{s.t. }~~~~\begin{split}&T_{l,k}(f_{l,k}^{(n)})+\frac{\beta d_k^2}{R_k(p_k,\Gamma_{r}^{(n)},\Gamma_t^{(n)},y_k^{(n)})} \\
&+T_{s,k}(f_{s,k}^{(n)})\leq t, \forall k\in\mathcal{K},\end{split}\\
& 0\leq p_k\leq p_{k,\text{max}}, \forall k\in\mathcal{K},\\
\begin{split}& E_{l,k}(f_{l,k}^{(n)})+\frac{\beta p_kd_k^2}{R_k(p_k,\Gamma_{r}^{(n)},\Gamma_t^{(n)},y_k^{(n)})} \\
&\leq E_{k,\text{max}},\forall k\in\mathcal{K},\end{split}
\end{align}
\end{subequations}
where $T_{l,k}(f_{l,k}^{(n)})=\frac{w_k}{f_{l,k}^{(n)}}$, $T_{s,k}(f_{s,k}^{(n)})=\frac{\beta d_k^2c_k}{f_{s,k}^{(n)}}$, $E_{l,k}(f_{l,k}^{(n)})=\kappa (f_{l,k}^{(n)})^2w_k$, and
\begin{equation}
\begin{aligned}
&R_k(p_k,\Gamma_{r}^{(n)},\Gamma_t^{(n)},y_k^{(n)})\!=\\
&\!\frac{B}{\ln 2}(\ln(\sum\limits_{k\in\mathcal{K}}p_k|h_{d,k}\!+\mathbf{h}_r^H\Gamma_{i(k)}^{(n)} \mathbf{h}_{I,k}|^2\!+\!B\delta^2)+\ln(y_k^{(n)})\\
&-y_k^{(n)}(\sum\limits_{j\neq k}p_j|h_{d,j}+\mathbf{h}_r^H\Gamma_{i(j)}^{(n)} \mathbf{h}_{I,j}|^2+B\delta^2)+1), \forall k\in\mathcal{K}.
\end{aligned}
\end{equation}
After the basic inequality transformations, (9) can be further rewritten as
\begin{subequations}
\begin{align}
\underset{\{p_{k}, t\}}{\text{minimize }}~& t\\
\text{s.t. }~\begin{split}&R_k(p_k,\Gamma_{r}^{(n)},\Gamma_t^{(n)},y_k^{(n)})\geq \\
&\frac{\beta d_k^2}{t-T_{l,k}(f_{l,k}^{(n)})-T_{s,k}(f_{s,k}^{(n)})}, \forall k\in\mathcal{K},\end{split}\\
& 0\leq p_k\leq p_{k,\text{max}}, \forall k\in\mathcal{K},\\
\begin{split}& R_k(p_k,\Gamma_{r}^{(n)},\Gamma_t^{(n)},y_k^{(n)})\geq\\
&\frac{\beta p_kd_k^2}{E_{k,\text{max}}-E_{l,k}(f_{l,k}^{(n)})}, \forall k\in\mathcal{K}.\end{split}
\end{align}
\end{subequations}

\begin{theorem}
Problem (11) is a convex optimization problem.
\end{theorem}

\emph{Proof: }To show the convexity, we first prove that $R_k(p_k,\Gamma_{r}^{(n)},\Gamma_t^{(n)},y_k^{(n)})$ is a concave function. According to convex optimization theory, $\ln(\sum\limits_{k\in\mathcal{K}}p_k|h_{d,k}+\mathbf{h}_r^H\Gamma_{i(k)}^{(n)} \mathbf{h}_{I,k}|^2+B\delta^2)$ is a concave function of $\{p_k\}$. Combined with other linear functions, we can derive that $R_k(p_k,\Gamma_{r}^{(n)},\Gamma_t^{(n)},y_k^{(n)})$ is a concave function associated with $\{p_k\}$. In addition, $\frac{\beta d_k^2}{t-T_{l,k}(f_{l,k}^{(n)})-T_{s,k}(f_{s,k}^{(n)})}$ is a convex function of $t$ \cite{boyd2004convex}. Therefore, (11b) and (11d) are convex constraints. Integrated with other linear objective function and constraints, we derive that problem (11) is a convex optimization problem.

Therefore, the standard convex optimization toolbox, such as CVX \cite{grant2014cvx}, can be called to solve the transmit power control problem (11).
\subsection{Transmit/Reflect Coefficient Matrix Optimization of STAR-RIS}
Let us define $\mathbf{v}_{i(k)}=(\sqrt{\gamma_{i(k),1}}e^{j\theta_{i(k),1}},\sqrt{\gamma_{i(k),2}}e^{j\theta_{i(k),2}},\cdots,\sqrt{\gamma_{i(k),N}}e^{j\theta_{i(k),N}})^T$, $\hat{\mathbf{v}}_{i(k)}=[\mathbf{v}_{i(k)}^T,1]^T$, $\mathbf{h}_k^H=[\mathbf{h}_r^H\text{diag}(\mathbf{h}_{I,k}),h_{d,k}]$, $\hat{\mathbf{V}}_{i(k)}=\hat{\mathbf{v}}_{i(k)}\hat{\mathbf{v}}_{i(k)}^H$, and $\mathbf{H}_k=\mathbf{h}_k\mathbf{h}_k^H$, it yields that
\begin{equation}
\begin{aligned}
&|h_{d,k}+\mathbf{h}_r^H\Gamma_{i(k)}\mathbf{h}_{I,k}|^2=\left|\mathbf{h}_k^H\hat{\mathbf{v}}_{i(k)}\right|^2\\
&=\mathbf{h}_k^H\hat{\mathbf{v}}_{i(k)}\hat{\mathbf{v}}_{i(k)}^H\mathbf{h}_k=\mathbf{h}_k^H\hat{\mathbf{V}}_{i(k)}\mathbf{h}_k\\
&=\text{Tr}(\hat{\mathbf{V}}_{i(k)}\mathbf{H}_k),\forall k\in\mathcal{K}.
\end{aligned}
\end{equation}
Therefore, for given $\{f_{l,k}^{(n)}, p_k^{(n)}, f_{s,k}^{(n)}\}$, (6) is transformed to the transmit/reflect coefficient matrix optimization subproblem as follows:
\begin{subequations}
\begin{align}
\underset{\{\hat{\mathbf{V}}_{t}, \hat{\mathbf{V}}_{r},\atop\gamma_{t,n},\gamma_{r,n},t\}}{\text{minimize }}~& t\\
\text{s.t. }~\begin{split}&T_{l,k}(f_{l,k}^{(n)})+\frac{\beta d_k^2}{R_k(p_k^{(n)},\hat{\mathbf{V}}_{t}, \hat{\mathbf{V}}_{r},y_k^{(n)})}\\
&+T_{s,k}(f_{s,k}^{(n)})\leq t, \forall k\in\mathcal{K},\end{split}\\
\begin{split}&E_{l,k}(f_{l,k}^{(n)})+\frac{\beta p_k^{(n)}d_k^2}{R_k(p_k^{(n)},\hat{\mathbf{V}}_{t}, \hat{\mathbf{V}}_{r},y_k^{(n)})}\\
&\leq E_{k,\text{max}},\forall k\in\mathcal{K},\end{split}\\
&[\hat{\mathbf{V}}_i]_{n,n}= \gamma_{i,n}, \forall i\in\{t,r\}, n\in\mathcal{N},\\
& [\hat{\mathbf{V}}_i]_{N+1,N+1}=1, \forall i\in\{t,r\},\\
&\hat{\mathbf{V}}_i\succeq 0, \forall i\in\{t,r\}, \\
&\text{Rank}(\hat{\mathbf{V}}_i)=1, \forall i\in\{t,r\},\\
& \gamma_{t,n}+\gamma_{r,n}=1, \forall n\in\mathcal{N},\\
& \gamma_{i,n}\geq 0, \forall i\in\{t,r\},n\in\mathcal{N},
\end{align}
\end{subequations}
where
\begin{equation}
\begin{aligned}
&R_k(p_k^{(n)},\hat{\mathbf{V}}_{t}, \hat{\mathbf{V}}_{r}, y_k^{(n)})=\frac{B}{\ln 2}(\ln(\sum\limits_{k\in\mathcal{K}}p_k^{(n)}\text{Tr}(\hat{\mathbf{V}}_{i(k)}\mathbf{H}_k)\!+B\delta^2)\\
&-y_k^{(n)}(\sum\limits_{j\neq k}p_j^{(n)}\text{Tr}(\hat{\mathbf{V}}_{i(j)}\mathbf{H}_j)+B\delta^2)+\ln(y_k^{(n)})+1), \forall k\in\mathcal{K}.
\end{aligned}
\end{equation}
After some simple inequality transformations, problem (13) can be recast as
\begin{subequations}
\begin{align}
\underset{\{\hat{\mathbf{V}}_{t}, \hat{\mathbf{V}}_{r},\atop\gamma_{t,n},\gamma_{r,n},t\}}{\text{minimize }}~& t\\
\text{s.t. }~\begin{split}&R_k(p_k^{(n)},\hat{\mathbf{V}}_{t}, \hat{\mathbf{V}}_{r}, y_k^{(n)})\!\geq \\
&\frac{\beta d_k^2}{t\!-T_{l,k}(f_{l,k}^{(n)})\!-T_{s,k}(f_{s,k}^{(n)})}, \forall k\in\mathcal{K},\end{split}\\
\begin{split}&R_k(p_k^{(n)},\hat{\mathbf{V}}_{t}, \hat{\mathbf{V}}_{r}, y_k^{(n)})\geq \\
&\frac{\beta p_k^{(n)}d_k^2}{E_{k,\text{max}}-E_{l,k}(f_{l,k}^{(n)})}, \forall k\in\mathcal{K},\end{split}\\
&\text{(13d)-(13i)}.
\end{align}
\end{subequations}
\vspace{-1em}

Above problem (15) is still non-convex due to the rank-one constraint. In general, the semi-definite relaxation technique can be utilized to drop non-convex rank-one constraint. However, it cannot ensure that the optimal solution is rank-one. Therefore, we propose a penalty function-based method to tackle the non-convex problem (15). As observed, an equivalent form for (13g) is expressed as $\text{Tr}(\hat{\mathbf{V}}_i)=\lambda_{\text{max}}(\hat{\mathbf{V}}_i)$ , where $\lambda_{\text{max}}(\hat{\mathbf{V}}_i)$ denotes the maximum eigenvalue of $\hat{\mathbf{V}}_i$. Moreover, the inequality $\text{Tr}(\hat{\mathbf{V}}_i)\geq \lambda_{\text{max}}(\hat{\mathbf{V}}_i)$ holds for any $\hat{\mathbf{V}}_i \succeq 0$, and the equality holds when $\text{Rank}(\hat{\mathbf{V}}_i)=1$. Hence, we can transform (15) to the following problem
\begin{subequations}
\begin{align}
\underset{\{\hat{\mathbf{V}}_{t}, \hat{\mathbf{V}}_{r},\gamma_{t,n},\gamma_{r,n},t\}}{\text{minimize }}&~~t+\nu\sum\limits_{i}(\text{Tr}(\hat{\mathbf{V}}_i)-\lambda_{\text{max}}(\hat{\mathbf{V}}_i))\\
\text{s.t. }~~~~
&\text{(13d)-(13f),(13h)-(13i),(15b)-(15c),}
\end{align}
\end{subequations}
where $\nu$ is the penalty factor, and rank of optimal transmit/reflect coefficient matrix for (16) will be forced to approach to 1 by enlarging the penalty factor. However, the convexity of $\lambda_{\text{max}}(\hat{\mathbf{V}}_i)$ makes the objective function (16a) non-convex. Based on the successive convex approximation method, $\lambda_{\text{max}}(\hat{\mathbf{V}}_i)$ is rewritten as its first-order Taylor approximation, i.e.,
\begin{equation}
\lambda_{\text{max}}(\hat{\mathbf{V}}_i)\geq \lambda_{\text{max}}(\hat{\mathbf{V}}^{(j)}_i)+(\hat{\mathbf{v}}_{\text{max},i}^{(j)})^H(\hat{\mathbf{V}}_i-\hat{\mathbf{V}}^{(j)}_i)\hat{\mathbf{v}}_{\text{max},i}^{(j)},
\end{equation}
where $\hat{\mathbf{V}}^{(j)}_i$ represents the optimal $\hat{\mathbf{V}}_i$ at the $j$-th iteration, and $\hat{\mathbf{v}}_{\text{max},i}^{(j)}$ denotes the unit eigenvector associate with maximum eigenvalue $\lambda_{\text{max}}(\hat{\mathbf{V}}^{(j)}_i)$ of $\hat{\mathbf{V}}^{(j)}_i$. Therefore, (16) is converted to the following problem
\begin{subequations}
\begin{align}
\begin{split}\underset{\{\hat{\mathbf{V}}_{t}, \hat{\mathbf{V}}_{r},\gamma_{r,n},\gamma_{t,n},t\}}{\text{minimize }}&~~t+\nu\sum\limits_i(\text{Tr}(\hat{\mathbf{V}}_i)-(\hat{\mathbf{v}}_{\text{max},i}^{(j)})^H\hat{\mathbf{V}}_i\hat{\mathbf{v}}_{\text{max},i}^{(j)})\end{split}\\
\text{s.t. }~~~~
&\text{(13d)-(13f),(13h)-(13i),(15b)-(15c).}
\end{align}
\end{subequations}
It should be noted that the constant term $\sum\limits_i (-\lambda_{\text{max}}(\hat{\mathbf{V}}^{(j)}_i)+(\hat{\mathbf{v}}_{\text{max},i}^{(j)})^H\hat{\mathbf{V}}^{(j)}_i\hat{\mathbf{v}}_{\text{max},i}^{(j)})$ is removed in (18a). Then, the classic convex toolbox, such as CVX, can be used to solve (18). After several number of iterations, the rank-one solution will be obtained when $\sum\limits_{i}(\text{Tr}(\hat{\mathbf{V}}_i^{(j)})- \lambda_{\text{max}}(\hat{\mathbf{V}}_i^{(j)}))\approx 0$. 
%Next, $\hat{\mathbf{v}}^{(j)}_i$ is acquired via eigenvalue decomposition, namely $\hat{\mathbf{v}}^{(j)}_i=\sqrt{\lambda_{\text{max}}(\hat{\mathbf{V}}^{(j)}_i)}\hat{\mathbf{v}}^{(j)}_{\text{max},i}$. Besides, the optimal coefficient matrix can be recovered by $\Gamma_i^{(j)}=\text{diag}(\hat{\mathbf{v}}^{(j)}_i(1:N))=\text{diag}(\sqrt{\lambda_{\text{max}}(\hat{\mathbf{V}}^{(j)}_i)}\hat{\mathbf{v}}^{(j)}_{\text{max},i}(1:N))$.

In summary, the proposed penalty function-based method for solving transmit/reflect coefficient matrix optimization subproblem is illustrated in Algorithm 1. Moreover, the convergence of Algorithm 1 is proved in the following \emph{Theorem 2}.
\begin{algorithm}[htbp]
\caption{Penalty function-based method for solving transmit/reflect coefficient matrix optimization subproblem (13)}
\textbf{Initialize:} Set $\{\hat{\mathbf{V}}_t^{(0)},\hat{\mathbf{V}}_r^{(0)},\gamma_{t,n}^{(0)},\gamma_{r,n}^{(0)}\}$, penalty factor $\nu$, and $j=1$.\\
\textbf{Repeat:}\\
\quad Acquire the optimal coefficient matrices $\hat{\mathbf{V}}_t^{(j)}$ and $\hat{\mathbf{V}}_r^{(j)}$ by solving (18);\\
%\quad Recovery $\Gamma_i^{(j)}=\text{diag}(\sqrt{\lambda_{\text{max}}(\hat{\mathbf{V}}^{(j)}_i)}\hat{\mathbf{v}}^{(j)}_{\text{max},i}(1:N))$;\\
\quad Update the iteration factor $j=j+1$;\\
\textbf{Until} convergence.\\
\textbf{Obtaining} optimal solution $\{\hat{\mathbf{V}}_t^{\ast}, \hat{\mathbf{V}}_r^{\ast}\}$ .\\
\end{algorithm}
\vspace{-1em}

%\begin{lemma}
%Assuming $\mathbf{Q}_1$ and $\mathbf{Q}_2$ as positive semi-definite matrixes. According to convex optimization theory, we have $\lambda_{\text{max}}(\mathbf{Q}_1)-\lambda_{\text{max}}(\mathbf{Q}_2)\geq \mathbf{q}_{\text{max}}^H(\mathbf{Q}_1-\mathbf{Q}_2)\mathbf{q}_{\text{max}}$, where $\mathbf{q}_{\text{max}}$ stands for eigenvector associate with maximum eigenvalue $\lambda_{\text{max}}(\mathbf{Q}_2)$ of $\mathbf{Q}_2$.
%\end{lemma}

\begin{theorem}
Algorithm 1 can converge to the optimal transmit/reflect coefficient matrix after a finite number of iterations.
\end{theorem}

\emph{Proof: }
Suppose that $\{\hat{\mathbf{V}}_t^{(j)}, \hat{\mathbf{V}}_r^{(j)}\}$ is the optimal solution of (18) at the $j$-th iteration. Given $\{f_{l,k}^{(n)}, p_k^{(n)}, f_{s,k}^{(n)}, y_k^{(n)}\}$, let us denote $t(\hat{\mathbf{V}}_t^{(j)}, \hat{\mathbf{V}}_r^{(j)})=\underset{k\in\mathcal{K}}{\text{max}} ~\{T_{l,k}(f_{l,k}^{(n)})+\frac{\beta d_k^2}{R_k(p_k^{(n)},\hat{\mathbf{V}}_{t}^{(j)}, \hat{\mathbf{V}}_{r}^{(j)},y_k^{(n)})}+T_{s,k}(f_{s,k}^{(n)})\}$. Therefore, we have
\begin{displaymath}
\begin{split}
&t(\hat{\mathbf{V}}_t^{(j+1)}, \hat{\mathbf{V}}_r^{(j+1)})+\nu\sum\limits_{i}(\text{Tr}(\hat{\mathbf{V}}_i^{(j+1)})-\lambda_{\text{max}}(\hat{\mathbf{V}}_i^{(j+1)}))\\
&\overset{(a1)}{\leq}t(\hat{\mathbf{V}}_t^{(j+1)}, \hat{\mathbf{V}}_r^{(j+1)})+\nu\sum\limits_{i}(\text{Tr}(\hat{\mathbf{V}}_i^{(j+1)})-\\
&\lambda_{\text{max}}(\hat{\mathbf{V}}^{(j)}_i)-(\hat{\mathbf{v}}_{\text{max},i}^{(j)})^H(\hat{\mathbf{V}}_i^{(j+1)}-\hat{\mathbf{V}}^{(j)}_i)\hat{\mathbf{v}}_{\text{max},i}^{(j)})\\
&\overset{(a2)}{\leq}t(\hat{\mathbf{V}}_t^{(j)}, \hat{\mathbf{V}}_r^{(j)})+\nu\sum\limits_{i}(\text{Tr}(\hat{\mathbf{V}}_i^{(j)})-\\
&\lambda_{\text{max}}(\hat{\mathbf{V}}^{(j)}_i)\!-\!(\hat{\mathbf{v}}_{\text{max},i}^{(j)})^H(\hat{\mathbf{V}}_i^{(j)}\!-\!\hat{\mathbf{V}}^{(j)}_i)\hat{\mathbf{v}}_{\text{max},i}^{(j)})\\
\end{split}
\end{displaymath}
\begin{equation}
\begin{split}
&=t(\hat{\mathbf{V}}_t^{(j)}, \hat{\mathbf{V}}_r^{(j)})+\nu\sum\limits_{i}(\text{Tr}(\hat{\mathbf{V}}_i^{(j)})\!-\!\lambda_{\text{max}}(\hat{\mathbf{V}}^{(j)}_i)),\\
\end{split}
\end{equation}
where $(a1)$ holds according to the feature of convex function $\lambda_{\text{max}}(\hat{\mathbf{V}}_i)$, i.e., $\lambda_{\text{max}}(\hat{\mathbf{V}}_i^{(j+1)})\geq \lambda_{\text{max}}(\hat{\mathbf{V}}^{(j)}_i)\!+(\hat{\mathbf{v}}_{\text{max},i}^{(j)})^H(\hat{\mathbf{V}}_i^{(j+1)}\!-\!\hat{\mathbf{V}}^{(j)}_i)\hat{\mathbf{v}}_{\text{max},i}^{(j)}$,
$(a2)$ holds because that the optimal $\{\hat{\mathbf{V}}_i^{(j+1)}\}$ is generated by minimizing (18a). Therefore, the optimal objective function (18a) is non-increasing after each iteration. Integrated with the boundedness of optimal solution of (18), we derive that the proposed penalty function-based method is able to converge a stationary point after a finite number of iterations.
\subsection{CPU-cycle Frequency Optimization of AR Users}
Given $\{p_{k}^{(n)}, \Gamma_t^{(n)}, \Gamma_r^{(n)}, f_{s,k}^{(n)}\}$, (6) is transformed as the following local CPU-cycle frequency optimization subproblem
\begin{subequations}
\begin{align}
\underset{\{f_{l,k}\}}{\text{minimize }}\begin{split}& \underset{k\in\mathcal{K}}{\text{max}}~\{\frac{w_k}{f_{l,k}}\!+\!T_{c,k}(p_k^{(n)},\!\Gamma_{r}^{(n)},\!\Gamma_t^{(n)},\!y_k^{(n)})\!+\!T_{s,k}(f_{s,k}^{(n)})\}\end{split}\\
\text{s.t. }~\begin{split}& \kappa w_kf_{l,k}^2+E_{c,k}(p_k^{(n)},\Gamma_{r}^{(n)},\Gamma_t^{(n)},y_k^{(n)})\\
&\leq E_{k,\text{max}},\forall k\in\mathcal{K},\end{split}\\
& 0\leq f_{l,k}\leq f_{k,\text{max}}, \forall k\in\mathcal{K},
\end{align}
\end{subequations}
where $T_{c,k}(p_k^{(n)},\Gamma_{r}^{(n)},\Gamma_t^{(n)},y_k^{(n)})=\frac{\beta d_k^2}{R_k(p_k^{(n)},\Gamma_{r}^{(n)},\Gamma_t^{(n)},y_k^{(n)})}$ and $E_{c,k}(p_k^{(n)},\Gamma_{r}^{(n)},\Gamma_t^{(n)},y_k^{(n)})=\frac{\beta p_k^{(n)} d_k^2}{R_k(p_k^{(n)},\Gamma_{r}^{(n)},\Gamma_t^{(n)},y_k^{(n)})}$. In the following \emph{Theorem 3}, we derive the optimal local CPU-cycle frequency as a closed-form solution for reducing the computational complexity.

\begin{theorem}
The optimal CPU-cycle frequency of AR users is expressed as
\begin{equation}
\begin{aligned}
&f_{l,k}^\ast=\text{min}\{\sqrt{\frac{E_{k,\text{max}}-E_{c,k}(p_k^{(n)},\Gamma_{r}^{(n)},\Gamma_t^{(n)},y_k^{(n)})}{\kappa w_k}},f_{k,\text{max}}\}, \\
&\forall k\in\mathcal{K}.
\end{aligned}
\end{equation}

\end{theorem}

\emph{Proof: }
After some conventional transformations, (20b) can be converted as
\begin{equation}
f_{l,k}\leq \sqrt{\frac{E_{k,\text{max}}-E_{c,k}(p_k^{(n)},\Gamma_{r}^{(n)},\Gamma_t^{(n)},y_k^{(n)})}{\kappa w_k}}, \forall k\in\mathcal{K}.
\end{equation}
Combining with (20c), we have
\begin{equation}
\begin{split}
&0\leq f_{l,k}\leq \text{min}\{\sqrt{\frac{E_{k,\text{max}}-E_{c,k}(p_k^{(n)},\Gamma_{r}^{(n)},\Gamma_t^{(n)},y_k^{(n)})}{\kappa w_k}},f_{k,\text{max}}\}, \\
&\forall k\in\mathcal{K}.
\end{split}
\end{equation}
As observed in (20a), the service latency of $k$-th AR user decreases with the increase of $f_{l,k}$. Therefore, the optimal CPU frequency of $k$-th AR user is obtained at the upper bound of $f_{l,k}$, as described in (21).
\subsection{Computation Resource Allocation of BS}
Given $\{f_{l,k}^{(n)}, p_{k}^{(n)}, \Gamma_t^{(n)}, \Gamma_r^{(n)}\}$, (6) is reduced to the BS's computation resource allocation subproblem
\begin{subequations}
\begin{align}
\underset{\{f_{s,k},t\}}{\text{minimize }}~&  t\\
\text{s.t. }~\begin{split}&T_{l,k}(f_{l,k}^{(n)})+T_{c,k}(p_k^{(n)},\Gamma_{r}^{(n)},\Gamma_t^{(n)},y_k^{(n)})\\
&+\frac{\beta d_k^2c_k}{f_{s,k}}\leq t, \forall k\in\mathcal{K},\end{split}\\
& \sum\limits_{k\in\mathcal{K}}f_{s,k}\leq F,\\
& f_{s,k}\geq 0, \forall k\in\mathcal{K}.
\end{align}
\end{subequations}

As observed, (24) is a convex problem due to the convexity of $1/f_{s,k}$. In addition, its feasible set is non-empty. Therefore, we adopt the Lagrange duality method and Karush-Kuhn-Tucker (KKT) conditions to analyze the optimal solution of (24). The Lagrange function of (24) is expressed as
\begin{equation}
\begin{aligned}
&L(f_{s,k},t,a_{1,k},a_2)=t+\sum\limits_{k\in\mathcal{K}}a_{1,k}(T_{l,k}(f_{l,k}^{(n)})+\\
&T_{c,k}(p_k^{(n)},\Gamma_{r}^{(n)},\Gamma_t^{(n)},y_k^{(n)})\!+\frac{\beta d_k^2c_k}{f_{s,k}}\!-t)\!+a_2(\sum\limits_{k\in\mathcal{K}}f_{s,k}\!-F),
\end{aligned}
\end{equation}
where $a_{1,k}$ and $a_2$ denote the non-negative Lagrange parameters associated with (24b) and (24c), respectively. The partial KKT conditions of (24) are given by
\begin{subequations}
\begin{align}
&\frac{\partial L(f_{s,k},t,a_{1,k},a_2)}{\partial f_{s,k}}=-a_{1,k}^\ast\frac{\beta d_k^2c_k}{(f_{s,k}^\ast)^2}+a_2^\ast=0, \forall k\in\mathcal{K},\\
&\frac{\partial L(f_{s,k},t,a_{1,k},a_2)}{\partial t} =1-\sum\limits_{k\in\mathcal{K}}a_{1,k}^\ast=0,\\
\begin{split}&a_{1,k}^\ast(T_{l,k}(f_{l,k}^{(n)})+T_{c,k}(p_k^{(n)},\Gamma_{r}^{(n)},\Gamma_t^{(n)},y_k^{(n)})+\\
&\frac{\beta d_k^2c_k}{f_{s,k}^\ast}-t^\ast)=0, \forall k\in\mathcal{K},\end{split}\\
&a_2^\ast(\sum\limits_{k\in\mathcal{K}}f_{s,k}^\ast- F)=0,\\
&a_{1,k}^\ast\geq 0, a_2^\ast\geq 0, \forall k\in\mathcal{K},
\end{align}
\end{subequations}
where (26c)-(26d) represent the complementary slackness conditions. According to (26a)-(26e), we will analyze the structure of optimal solution of (24) in the following theorems to obtain interesting insights.

\begin{theorem}
For the computation resource allocation problem (24), the optimal solution should satisfy $\sum\limits_{k\in\mathcal{K}}f_{s,k}^\ast= F$.
\end{theorem}

\emph{Proof: }
We first prove $a_{1,k}^\ast>0, \forall k\in\mathcal{K}$, by contradiction. Assuming that $\{a_{1,k}^\ast, a_2^\ast,f_{s,k}^\ast,t^\ast\}$ is the optimal solution of (26), and there exists $a_{1,m}^\ast=0, m\in\mathcal{K}$. According to (26a), it yields that $a_{1,m}^\ast\frac{\beta d_m^2c_m}{(f_{s,m}^\ast)^2}=0=a_2^\ast=a_{1,k}^\ast\frac{\beta d_k^2c_k}{(f_{s,k}^\ast)^2}, \forall k\in\mathcal{K}$. Therefore, we have $a_{1,k}^\ast=0, \forall k\in\mathcal{K}$, which is contradicted with $1-\sum\limits_{k\in\mathcal{K}}a_{1,k}^\ast=0$. Hence, we can derive that $a_{1,k}^\ast>0, \forall k\in\mathcal{K}$, and $a_2^\ast=a_{1,k}^\ast\frac{\beta d_k^2c_k}{(f_{s,k}^\ast)^2}>0$. As observed in (26d), we have $\sum\limits_{k\in\mathcal{K}}f_{s,k}^\ast= F$.

\begin{theorem}
For the optimal solution of computation resource allocation problem (24), each AR user will enjoy the same service latency, i.e.,
\begin{equation}
\begin{aligned}
&T_{l,k}(f_{l,k}^{(n)})+T_{c,k}(p_k^{(n)},\Gamma_{r}^{(n)},\Gamma_t^{(n)},y_k^{(n)})+T_{s,k}(f_{s,k}^{\ast})=\\
&T_{l,m}(f_{l,m}^{(n)})+T_{c,m}(p_m^{(n)},\Gamma_{r}^{(n)},\Gamma_t^{(n)},y_m^{(n)})+T_{s,m}(f_{s,m}^\ast), \\
&\forall k,m\in\mathcal{K}, k\neq m.
\end{aligned}
\end{equation}
\end{theorem}

\emph{Proof}
According to the proof of Theorem 4, we have $a_{1,k}^\ast>0, \forall k\in\mathcal{K}$. Combined with (26c), we have $T_{l,k}(f_{l,k}^{(n)})+T_{c,k}(p_k^{(n)},\Gamma_{r}^{(n)},\Gamma_t^{(n)},y_k^{(n)})+\frac{\beta d_k^2c_k}{f_{s,k}^\ast}=t^\ast, \forall k\in\mathcal{K}$. Therefore, for the optimal solution of (24), each AR user has the same service latency $t^\ast$.

\emph{Remark 5: }As shown in (27), all the AR users has the same service latency under the optimal resource allocation strategy. Therefore, the proposed method is able to ensure the fairness of AR users in terms of service latency. In addition, the BS tends to allocate more computation resources to analyze (i.e., object detection) the image from the AR user with longer local processing and transmission latency.

\begin{theorem}
The optimal computation resource allocation strategy should satisfy
\begin{equation}
f_{s,k}^\ast=d_k\sqrt{\frac{a_{1,k}^\ast\beta c_k}{a_{2}^\ast}}, \forall k\in\mathcal{K}.
\end{equation}
\end{theorem}

\emph{Proof}
Based on (26a), it yields that
\begin{equation}
\begin{aligned}
-a_{1,k}^\ast\frac{\beta d_k^2c_k}{(f_{s,k}^\ast)^2}+a_2^\ast=0, \forall k\in\mathcal{K}.
\end{aligned}
\end{equation}
Therefore, we can derive the optimal computation resource assignment strategy as given in (28).

\emph{Remark 6: }As shown in (28), the optimal computation resource assigned to $k$-th AR user, namely $f_{s,k}^\ast$, increases with the image size $d_k$ and the computational complexity $c_k$. This implies that the BS tends to allocate higher computing frequency to the AR user with higher computation requirements. Moreover, we also find that, as compared to the computational complexity $c_k$, the image size $d_k$ exhibits a larger impact on the computation resource strategy of the BS, i.e., $f_{s,k}^\ast\sim d_k$ and $f_{s,k}^\ast\sim \sqrt{c_k}$.
\subsection{Updation of Auxiliary Variable $y_k$}
Given $\{p_k^{(n)}, \Gamma_{t}^{(n)}, \Gamma_{r}^{(n)}, f_{l,k}^{(n)}, f_{s,k}^{(n)}\}$, the auxiliary variable $y_k$ is updated according to \emph{Lemma 1} as below
\begin{equation}
y_k^{(n)}=\frac{1}{\sum\limits_{j\neq k}p_j^{(n)}|h_{d,j}+\mathbf{h}_r^H\Gamma_{i(j)}^{(n)} \mathbf{h}_{I,j}|^2+B\delta^2}, \forall k\in\mathcal{K}.
\end{equation}
\subsection{Algorithm, Convergence, and Complexity}
As described above, the proposed method is summarized at the following Algorithm 2. Besides, convergence and computational complexity are analyzed as follows.
\begin{algorithm}[htbp]
\caption{Proposed method for service latency minimization problem (5)}
\textbf{Initialize:} Setting $\{p_k^{(0)}, \hat{\mathbf{V}}_{t}^{(0)}, \hat{\mathbf{V}}_{r}^{(0)}, f_{l,k}^{(0)}, f_{s,k}^{(0)}, y_{k}^{(0)}\}$, and $n=1$.\\
\textbf{Repeat:}\\
\quad Given $\{\hat{\mathbf{V}}_{t}^{(n-1)}, \hat{\mathbf{V}}_{r}^{(n-1)}, f_{l,k}^{(n-1)}, f_{s,k}^{(n-1)}, y_{k}^{(n-1)}\}$, calculating the optimal power control strategy $\{p_k^{(n)}\}$ by solving (11);\\
\quad Given $\{p_k^{(n)}, f_{l,k}^{(n-1)}, f_{s,k}^{(n-1)}, y_{k}^{(n-1)}\}$, executing Algorithm 1 to acquire the optimal coefficient matrixes $\{\hat{\mathbf{V}}_t^{(n)},\hat{\mathbf{V}}_r^{(n)}\}$;\\
\quad Given $\{p_k^{(n)}, \hat{\mathbf{V}}_{t}^{(n)}, \hat{\mathbf{V}}_{r}^{(n)},  f_{s,k}^{(n-1)}, y_{k}^{(n-1)}\}$, obtaining the optimal local computing frequencies $\{f_{l,k}^{(n)}\}$ via (21);\\
\quad Given $\{p_k^{(n)}, \hat{\mathbf{V}}_{t}^{(n)}, \hat{\mathbf{V}}_{r}^{(n)}, f_{l,k}^{(n)}, y_{k}^{(n-1)}\}$, acquiring the optimal computation resource allocation strategy $\{f_{s,k}^{(n)}\}$ by solving (24);\\
\quad Given $\{p_k^{(n)}, \hat{\mathbf{V}}_{t}^{(n)}, \hat{\mathbf{V}}_{r}^{(n)}, f_{l,k}^{(n)}, f_{s,k}^{(n)}\}$, computing $y_k^{(n)}=\frac{1}{\sum\limits_{j\neq k}p_j^{(n)}|h_{d,j}+\mathbf{h}_r^H\Gamma_{i(j)}^{(n)} \mathbf{h}_{I,j}|^2+B\delta^2}, \forall k\in\mathcal{K}$;\\
\quad Updating iteration factor $n=n+1$;\\
\textbf{Until} convergence.\\
\textbf{Return:} the optimal solution $\{p_k^{\ast},\hat{\mathbf{V}} _t^{\ast}, \hat{\mathbf{V}}_r^\ast, f_{l,k}^\ast,  f_{s,k}^{\ast}\}$.
\end{algorithm}

\begin{theorem}
The proposed alternating optimization method is able to converge to the optimal solution within a finite number of iterations.
\end{theorem}

\emph{Proof: }
The proof is similar to the convergence analysis in \cite{10360849}.

Next, we analyze the computational complexity of Algorithm 2. Defining $N_1$ and $N_2$ as the iteration numbers of Algorithm 1 and Algorithm 2, respectively. In each iteration of Algorithm 2, four subproblems are alternately solved with the computational complexity given as follows.
\begin{itemize}
\item To obtain the optimal transmit power of AR users, the interior-point method is utilized to solve the convex problem (11) with $K+1$ variables and $4K$ constraints. The computational complexity is expressed as $\mathcal{O}(2(5K+1)(K+1)^2\sqrt{K}\log(\frac{1}{\epsilon_1}))$, where $\epsilon_1$ is the tolerance factor.
\item To acquire the optimal transmit/reflect coefficient matrix, Algorithm 1 is adopted to solve the non-convex problem (13). In each iteration, Algorithm 1 needs to iteratively solve convex problem (18) with $2(N+1)^2+1$ variables and $5N+2K+4$ constraints. The corresponding computational complexity is expressed as $\mathcal{O}((2(N+1)^2+5N+2K+5)(2(N+1)^2+1)^2\sqrt{5N+2K+4}\log(\frac{1}{\epsilon_2}))$ with $\epsilon_2$ be the tolerance factor. Combined with the iteration number $N_1$, the computational complexity of Algorithm 1 is $\mathcal{O}(N_1(2(N+1)^2+5N+2K+5)(2(N+1)^2+1)^2\sqrt{5N+2K+4}\log(\frac{1}{\epsilon_2}))$.
\item To obtain the optimal CPU frequencies of AR users, we derive the optimal solution in closed-form expressions as given in (21). Therefore, the computational complexity is $\mathcal{O}(K)$.
\item To acquire the optimal computation resource allocation strategy, we need to solve the convex problem (24) with $K+1$ variables and $2K+1$ constraints. Thus, the computational complexity is expressed as $\mathcal{O}((3K+2)(K+1)^2\sqrt{2K+1}\log(\frac{1}{\epsilon_3}))$ with tolerance factor $\epsilon_3$.
\end{itemize}
Moreover, the complexity for updating auxiliary variable $y_k$ is given by $\mathcal{O}(K)$. Therefore, the total computational complexity of Algorithm 2 is expressed as $\mathcal{O}(N_2(2(5K+1)(K+1)^2\sqrt{K}\log(\frac{1}{\epsilon_1})+N_1(2(N+1)^2+5N+2K+5)(2(N+1)^2+1)^2\sqrt{5N+2K+4}\log(\frac{1}{\epsilon_2})+2K+(3K+2)(K+1)^2\sqrt{2K+1}\log(\frac{1}{\epsilon_3})))$.

According to above descriptions, Algorithms 2 is proposed to solve the non-convex service latency minimization problem with polynomial complexity order. In contrast, defining $G_v$ as the candidate value size for each optimization variable, the complexity of exhaustive search algorithm is $\mathcal{O}(G_v^{3K+2M})$, which exponentially increases as $K$ and $M$.  Moreover, the particle swarm optimization algorithm can be adjusted to solve the non-convex optimization problem with the complexity of $\mathcal{O}(I_{t}I_{s}I_{f})$, where $I_{t}$ denotes the iteration number, $I_{s}$ represents the swarm size, and $I_f$ stands for the complexity in evaluating the fitness value in formulated optimization problems. Denoting $O_v$ as the size of optimization variables for formulated problems, we have $I_f\sim O_v^{3.5}$, which is close to the per-iteration complexity of proposed method. In addition, the swarm size $I_{s}$ and iteration number $I_{t}$ increase with the search space dimension and problem complexity remarkably. Considering the high-dimensional optimization variables, the proposed method has lower computational complexity than the particle swarm optimization-based algorithms.
\section{Service Latency Minimization Problem for FDMA-based Systems}
In this section, we investigate the service latency minimization problem for FDMA-based Metaverse systems. For FDMA-based systems, the transmission rate of $k$-th AR user is given by
\begin{equation}
\begin{aligned}
&R_{k}^{\text{FDMA}}\!=\!b_kB\log_2\left(1\!+\frac{p_k|h_{d,k}\!+\mathbf{h}_r^H\Gamma_{i(k)} \mathbf{h}_{I,k}|^2}{b_kB\delta^2}\right), \\
&\forall k\in\mathcal{K},
\end{aligned}
\end{equation}
where $b_k\in[0,1]$ indicates the bandwidth allocation ratio for $k$-th AR user. Similar to the method proposed in Section III, the service latency minimization problem for FDMA-based systems can be  addressed by alternately solving several subproblems, including transmit power and bandwidth allocation subproblem, transmit/reflect coefficient matrix optimization subproblem, CPU-cycle frequency optimization subproblem, and computation resource allocation subproblem. In addition to the transmit power and bandwidth allocation subproblem, the other subproblems can be solved by the same methods in Section III, through replacing $R_k$ as $R_{k}^{\text{FDMA}}$. Under given $\{f_{l,k}^{(n)}, f_{s,k}^{(n)}, \Gamma_t^{(n)}, \Gamma_r^{(n)}\}$, the transmit power and bandwidth allocation subproblem can be expressed as
\begin{subequations}
\begin{align}
\underset{\{b_k, p_{k}, t\}}{\text{minimize }}~~~& t\\
\text{s.t. }~~~~\begin{split}&T_{l,k}(f_{l,k}^{(n)})+\frac{\beta d_k^2}{R_k^{\text{FDMA}}(b_k,p_k,\Gamma_{r}^{(n)},\Gamma_t^{(n)})} \\
&+T_{s,k}(f_{s,k}^{(n)})\leq t, \forall k\in\mathcal{K},\end{split}\\
& 0\leq p_k\leq p_{k,\text{max}}, \forall k\in\mathcal{K},\\
\begin{split}& E_{l,k}(f_{l,k}^{(n)})+\frac{\beta p_kd_k^2}{R_k^{\text{FDMA}}(b_k,p_k,\Gamma_{r}^{(n)},\Gamma_t^{(n)})}\\
&\leq E_{k,\text{max}}, \forall k\in\mathcal{K},\end{split}\\
& \sum\limits_{k\in\mathcal{K}} b_k \leq 1,\\
& b_k\geq 0, \forall k\in\mathcal{K},
\end{align}
\end{subequations}
where $T_{l,k}(f_{l,k}^{(n)})=\frac{w_k}{f_{l,k}^{(n)}}$, $T_{s,k}(f_{s,k}^{(n)})=\frac{\beta d_k^2c_k}{f_{s,k}^{(n)}}$, $E_{l,k}(f_{l,k}^{(n)})=\kappa (f_{l,k}^{(n)})^2w_k$, and
\begin{equation}
\begin{aligned}
&R_k^{\text{FDMA}}(b_k,p_k,\Gamma_{r}^{(n)},\Gamma_t^{(n)})\!=\\
&b_kB\log_2\left(1\!+\frac{p_k|h_{d,k}\!+\mathbf{h}_r^H\Gamma_{i(k)}^{(n)} \mathbf{h}_{I,k}|^2}{b_kB\delta^2}\right) , \forall k\in\mathcal{K}.
\end{aligned}
\end{equation}
After basic inequality transformations, (32) can be converted as
\begin{subequations}
\begin{align}
\underset{\{b_k, p_{k}, t\}}{\text{minimize }}~& t\\
\text{s.t. }~\begin{split}&R_k^{\text{FDMA}}(b_k, p_k,\Gamma_{r}^{(n)},\Gamma_t^{(n)})\geq \\
&\frac{\beta d_k^2}{t-T_{l,k}(f_{l,k}^{(n)})-T_{s,k}(f_{s,k}^{(n)})}, \forall k\in\mathcal{K},\end{split}\\
\begin{split}& R_k^{\text{FDMA}}(b_k, p_k,\Gamma_{r}^{(n)},\Gamma_t^{(n)})\geq \\
&\frac{\beta p_kd_k^2}{E_{k,\text{max}}-E_{l,k}(f_{l,k}^{(n)})}, \forall k\in\mathcal{K},\end{split}\\
&\text{(32c), (32e)-(32f).}
\end{align}
\end{subequations}
Since $\log_2\left(1\!+\frac{p_k|h_{d,k}\!+\mathbf{h}_r^H\Gamma_{i(k)}^{(n)} \mathbf{h}_{I,k}|^2}{\delta^2}\right)$ is a typically concave function of $p_k$, its perspective function $R_k^{\text{FDMA}}(b_k, p_k,\Gamma_{r}^{(n)},\Gamma_t^{(n)},y_k^{(n)})$ with $b_k\geq 0$ is also concave associated with $b_k$ and $p_k$. Therefore, the transmit power and bandwidth allocation subproblem can be proved as a convex problem, which can be solved by classic convex optimization methods integrated on CVX.

\begin{algorithm}[htbp]
\caption{Proposed method for solving service latency minimization problem (31) of FDMA-based Metaverse systems}
\textbf{Initialize:} Setting $\{p_k^{(0)}, b_k^{(0)}, \hat{\mathbf{V}}_{t}^{(0)}, \hat{\mathbf{V}}_{r}^{(0)}, f_{l,k}^{(0)}, f_{s,k}^{(0)}\}$, and $n=1$.\\
\textbf{Repeat:}\\
\quad Given $\{\hat{\mathbf{V}}_{t}^{(n-1)}, \hat{\mathbf{V}}^{(n-1)}, f_{l,k}^{(n-1)}, f_{s,k}^{(n-1)}\}$, obtaining the optimal transmit power control and bandwidth allocation strategy $\{p_k^{(n)}, b_k^{(n)}\}$ by solving (34);\\
\quad Given $\{p_k^{(n)}, b_k^{(n)}\}$, utilizing the similar methods of Algorithm 2 to acquire the optimal $\{\hat{\mathbf{V}}_{t}^{(n)},\hat{\mathbf{V}}_{r}^{(n)},f_{l,k}^{(n)},f_{s,k}^{(n)}\}$, via replaying $R_k$ as $R_k^{\text{FDMA}}$;\\
\quad Updating iteration factor $n=n+1$;\\
\textbf{Until} convergence.\\
\textbf{Return:} the optimal solution $\{p_k^{\ast}, b_k^{\ast}, \hat{\mathbf{V}} _t^{\ast}, \hat{\mathbf{V}}_r^\ast, f_{l,k}^\ast,  f_{s,k}^{\ast}\}$.
\end{algorithm}

In Algorithm 3, we summarize the proposed method to solve the service latency minimization problem (31) for FDMA-based STAR-RIS-enhanced Metaverse systems. Then, we summarize the computational complexity of Algorithm 3. The interior-point method-based solver is exploited to solve the transmit power and bandwidth allocation subproblem with $2K+1$ variables and $5K+1$ constraints, so its computational complexity is given by $\mathcal{O}((2K+1)^2(7K+2)\sqrt{5K+1}\log(\frac{1}{\epsilon_4}))$, where $\epsilon_4$ indicates the tolerance factor. Moreover, the same methods in Section III are utilized to solve the other subproblems by replacing $R_k$ as $R_k^{\text{FDMA}}$. Defining $N_3$ as the iteration number of Algorithm 3, its total computational complexity is expressed as $\mathcal{O}(N_3((2K+1)^2(7K+2)\sqrt{5K+1}\log(\frac{1}{\epsilon_4})+N_1(2(N+1)^2+5N+2K+5)(2(N+1)^2+1)^2\sqrt{5N+2K+4}\log(\frac{1}{\epsilon_2})+(3K+2)(K+1)^2\sqrt{2K+1}\log(\frac{1}{\epsilon_3})+K))$.

It is worth noting that this article considers the scenario with single-antenna BS, and the proposed methods can be extended to the scenario with multi-antenna BS after some simple modifications. The similar methods in the single-antenna BS scenario can be utilized to obtain the optimal transmit power, bandwidth allocation ratios, CPU-cycle frequencies of AR users, transmit/reflect coefficient matrix of STAR-RIS, and computation resource allocation of BS. Meanwhile, the receive beamforming strategy can be achieved by exploiting the similar variable substitution technique and penalty function-based method of Section III-B.

\section{Numerical Results}
This section provides extensive numerical results to evaluate the achievable service latency of proposed method for STAR-RIS-assisted AR-empowered Metaverse systems, by comparing with the following several benchmark methods\footnote{In order to ensure the fairness in comparison, the simulation parameters of benchmark methods are same to our proposed methods.}:
\begin{itemize}
\item \emph{Reflecting RIS: }In this benchmark scheme, a conventional reflecting-only RIS is employed to assist the uplink data transmission of AR users located in the reflection region, i.e., $\gamma_{r,n}=1$ and $\gamma_{t,n}=0, \forall n\in\mathcal{N}$. To this end, the transmit power and local CPU-cycle frequency of AR users, the RIS reflection coefficient matrix, and the computation resource allocation at the BS are jointly optimized to minimize the service latency.
\item \emph{Transmitting RIS: }In this scheme, the RIS operates in a refracting-only mode to support the uplink data transmission of AR users located in the transmission region, i.e., $\gamma_{r,n}=0$ and $\gamma_{t,n}=1, \forall n\in\mathcal{N}$. The optimization framework jointly considers the transmit coefficient matrix and resource management strategy in order to reduce the service latency.
\item \emph{STAR-RIS with random phase: }In this scheme, the reflection and transmission amplitudes of each RIS element are fixed as $\gamma_{r,n}=\gamma_{t,n}=0.5, \forall n\in\mathcal{N}$, while the corresponding phase shifts are randomly drawn from the interval $[0, 2\pi]$. Meanwhile, the proposed optimization method is applied to jointly allocate communication and computing resources, aiming to minimize the maximum service latency across all AR users.
\end{itemize}

In the simulation, we consider a scenario with $K_r=5$ reflection AR users and $K_t=5$ transmission AR user, and the STAR-RIS is equipped with $N=16$ elements. Moreover, the coordinates of BS and STAR-RIS are set as [0, 10] and [10, 0],  respectively. In addition, AR users in the reflection region are uniformly distributed within a square area defined by the x-coordinate range [8, 10] and the y-coordinate range [0, 2], while those in the transmission region are uniformly distributed within a square area defined by the x-coordinate range [10, 12] and the y-coordinate range [0, 2]. Furthermore, the other simulation parameters are set as follows: $\rho=-30$ dB, $\iota_{l,k}=\iota_{r}=2$, $\iota_{d,k}=3$, $\varsigma=3$, $\kappa=10^{-27}$, $\nu=10^{-5}$, $\delta^2=10^{-13}$ W/Hz, $E_{k,\text{max}}=2$ Joule, $f_{k,\text{max}}=1$ GHz, $c_k=1000$, $w_k=5\times 10^7$ Cycles, and $\beta=25$ \cite{feng2022resource}.

\begin{figure}[htbp]
\begin{center}
\includegraphics[width=0.5\textwidth]{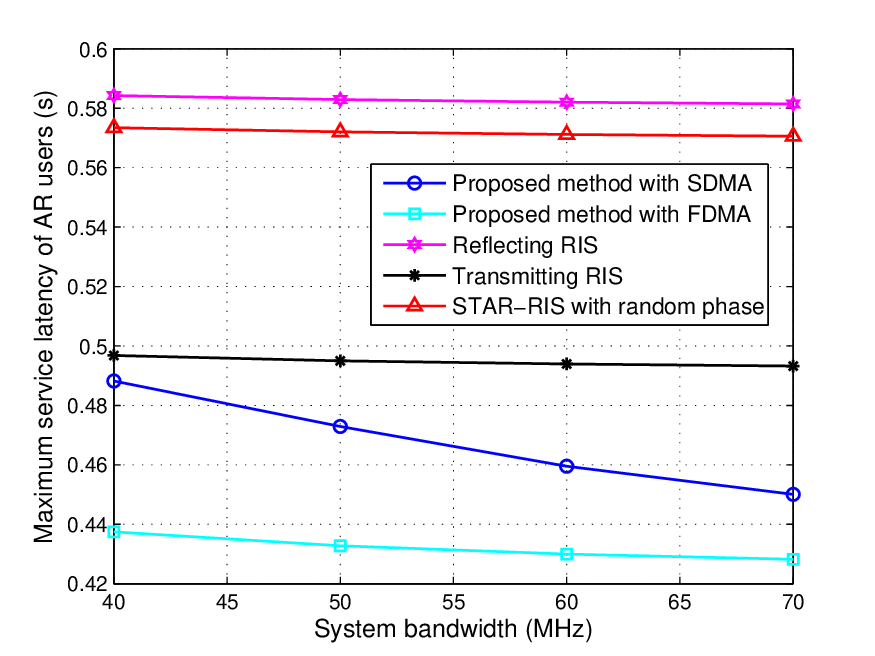}
\caption{Maximum service latency of AR users versus system bandwidth.}
\end{center}
\end{figure}

\begin{figure}[htbp]
	\begin{center}
		\includegraphics[width=0.5\textwidth]{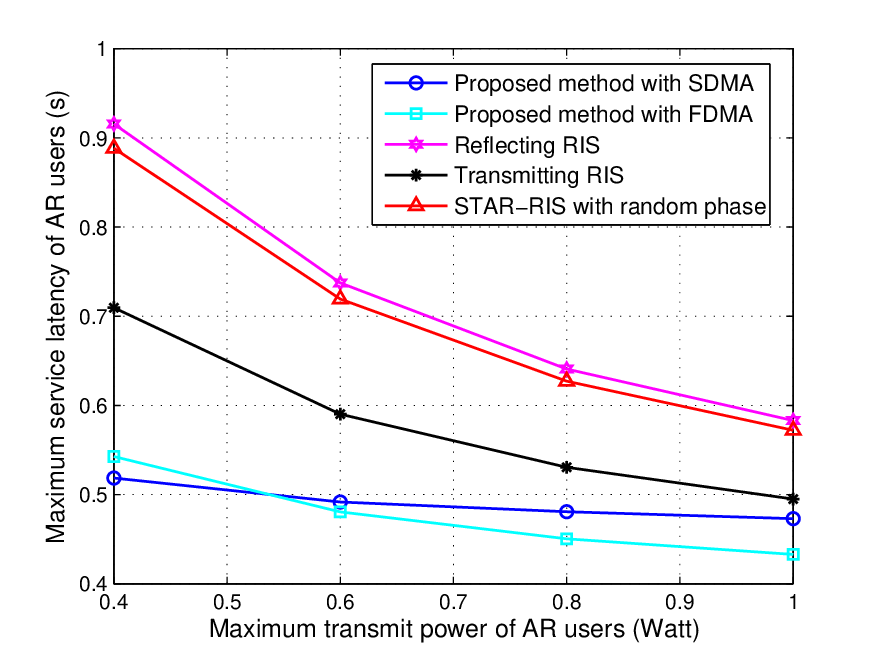}
		\caption{Maximum service latency of AR users versus maximum transmit power of AR users.}
	\end{center}
\end{figure}

In Fig.~2, we present the maximum service latency of AR users versus the system bandwidth $B$, where the maximum transmit power of AR users is set to $p_{k,\text{max}}=1$~W, the computation model size is $d_k=150$~pixels, the number of STAR-RIS elements is $N=16$, and the maximum computing frequency of the BS is set to $F=20$~GHz. As expected, the maximum service latency achieved by all schemes decreases as the system bandwidth increases. Moreover, the decreasing trend becomes less pronounced with larger bandwidth values, since the communication latency of AR users is inversely proportional to the transmission bandwidth. In addition, it is observed that the proposed method achieves substantial latency reduction compared with several benchmark schemes, particularly when the system bandwidth is small. This observation indicates that the optimal deployment and configuration of STAR-RIS is especially beneficial for AR-empowered Metaverse systems operating under limited communication resources. Furthermore, as shown in the figure, the FDMA-based scheme yields lower service latency than the SDMA-based scheme, although the performance gap narrows with increasing system bandwidth. This behavior can be explained by the harmful co-channel interference inherent in the SDMA-based scheme, which significantly degrades the communication latency between AR users and the BS. By contrast, the FDMA-based scheme completely eliminates such interference through orthogonal frequency allocation. Nevertheless, as the system bandwidth expands, both FDMA- and SDMA-based schemes achieve low communication latency, thereby reducing the relative performance advantage of FDMA.  

Fig.~3 illustrates the relationship between the maximum service latency of AR users and their maximum transmit power $p_{k,\text{max}}$, where $B=50$~MHz, $d_k=150$~pixels, $F=20$~GHz, and $N=16$. As anticipated, the maximum service latency monotonically decreases with increasing $p_{k,\text{max}}$, since higher transmit power improves the uplink transmission rate of AR users, thereby reducing latency. It is observed that when $p_{k,\text{max}}=0.2$~W, the SDMA-based scheme outperforms the FDMA-based scheme in terms of latency. In contrast, when $p_{k,\text{max}}\geq 0.4$~W, the FDMA-based scheme demonstrates superior performance over its SDMA-based counterpart. It reveals that SDMA is more advantageous in power-constrained scenarios, whereas FDMA exhibits superior performance when the AR users are equipped with relatively higher transmission power. Moreover, the proposed method consistently achieves significantly lower service latency compared with all benchmark schemes, particularly in the low transmit power regime. This is because with higher $p_{k,\text{max}}$, AR users can effectively mitigate path loss in uplink transmissions by increasing their transmit power, thereby reducing the relative performance gains provided by the optimal deployment and configuration of STAR-RIS.

\begin{figure}[htbp]
	\begin{center}
		\includegraphics[width=0.5\textwidth]{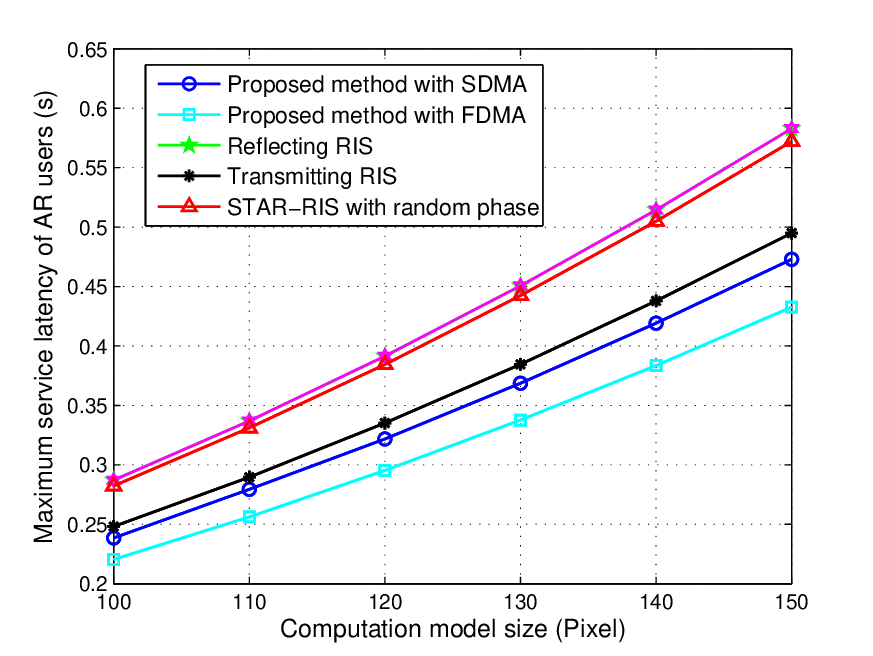}
		\caption{Maximum service latency of AR users versus computation model size.}
	\end{center}
\end{figure}
Fig. 4 plots the maximum service latency of AR users versus the computation model size $d_k$, where $B=50$ MHz, $p_{k,\text{max}}=1$ W, $F=20$ GHz, and $N=16$. As anticipated, the maximum service latency increases with $d_k$, since both transmission latency and edge processing latency inevitably grow with larger computation model sizes. A notable observation is that the latency gap between the proposed method and the benchmark schemes widens as $d_k$ increases, indicating that the optimal configuration of the STAR-RIS is particularly effective in mitigating transmission delay when the uplink data volume is large. Furthermore, the FDMA-based scheme consistently achieves lower service latency than its SDMA-based counterpart, with the performance gap becoming more pronounced as $d_k$ grows. This observation confirms that FDMA provides more substantial latency reduction under heavy communication demands, where severe co-channel interference significantly degrades the performance of SDMA-based schemes.

\begin{figure}[htbp]
	\begin{center}
		\includegraphics[width=0.5\textwidth]{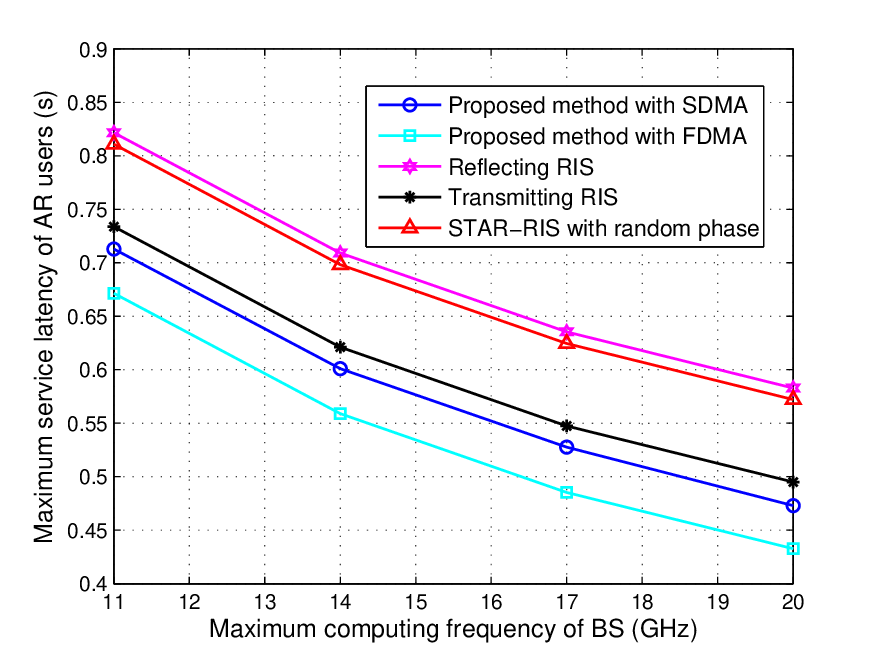}
		\caption{Maximum service latency of AR users versus maximum computing frequency of BS.}
	\end{center}
\end{figure}
Fig. 5 shows the maximum service latency of AR users against maximum computing frequency $F$ of the BS, where $B=50$ MHz, $p_{k,\text{max}}=1$ W,  $d_k=150$ pixels, and $N=16$. It can be observed from the figure that the maximum service latency among all AR users decreases as the maximum computing frequency of the BS increases, owing to the reduced object detection latency enabled by the enhanced computational capability at the BS. Furthermore, the proposed method achieves latency reductions of at least 18\%, 8\%, and 17\% compared to the \emph{Reflecting RIS} scheme, the \emph{Transmitting RIS} scheme, and the \emph{STAR-RIS with random phase} scheme, respectively.

\begin{figure}[htbp]
	\begin{center}
		\includegraphics[width=0.5\textwidth]{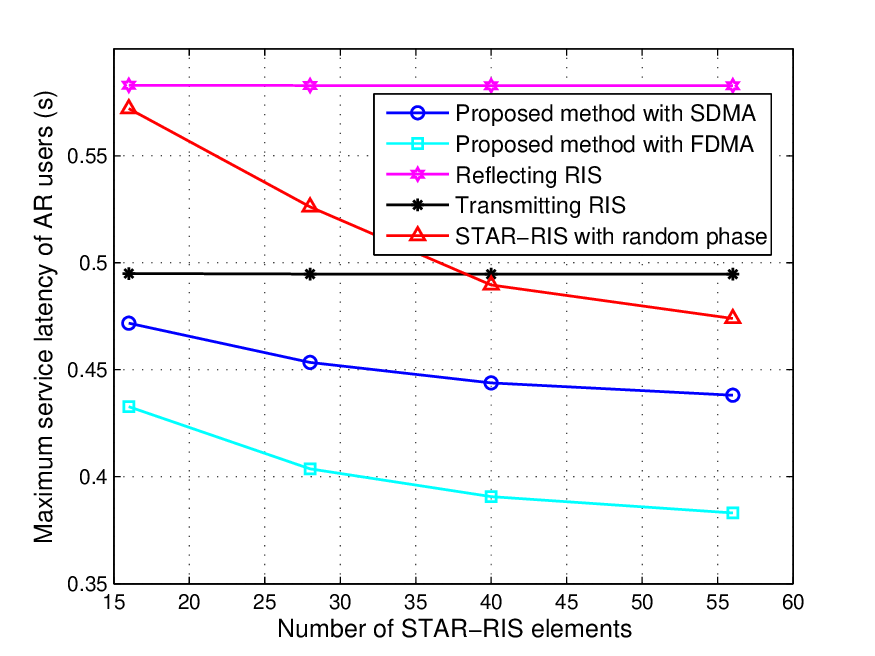}
		\caption{Maximum service latency of AR users versus number of STAR-RIS elements.}
	\end{center}
\end{figure}
Fig.~6 illustrates the maximum service latency of AR users versus the number of STAR-RIS elements $N$, where $B=50$ MHz, $p_{k,\text{max}}=1$ W,  $d_k=150$ pixels, and $F=20$ GHz. As $N$ increases, the service latency of AR users achieved by the proposed method decreases substantially, since the additional STAR-RIS elements provide extra spatial degrees of freedom to enhance the data transmission rate of AR users, thereby reducing the overall service latency. Moreover, the proposed method attains a more pronounced latency reduction compared to both the \emph{Reflecting RIS} scheme and the \emph{Transmitting RIS} scheme, as these benchmark schemes can only assist the uplink data transmission of AR users located in either the transmission or reflection region. Consequently, merely increasing the number of RIS elements is insufficient to effectively improve the worst-case service latency. In addition, it is observed that the FDMA-based method achieves significantly lower service latency than the SDMA-based counterpart, with the performance gap widening as the number of STAR-RIS elements increases. This is because FDMA avoids co-channel interference, allowing the performance improvement from additional STAR-RIS elements to be more fully realized. In contrast, the SDMA-based method still suffers from harmful co-channel interference, which limits the performance gain from increasing the number of STAR-RIS elements.

\begin{figure}[htbp]
	\begin{center}
		\includegraphics[width=0.5\textwidth]{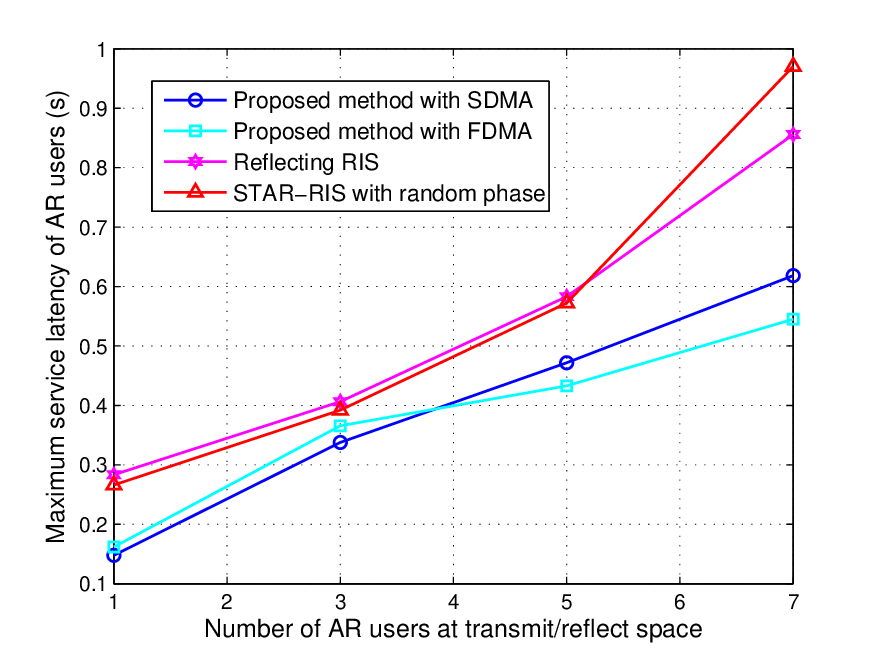}
		\caption{Maximum service latency of AR users versus number of AR users.}
	\end{center}
\end{figure}
Fig.~7 illustrates the impact of the number of AR users $K_t/K_r$ on the maximum service latency across all AR users, where $B=50$ MHz, $p_{k,\text{max}}=1$ W,  $d_k=150$ pixels, $F=20$ GHz, and $N=16$. As $K_t/K_r$ increases, both communication latency and edge processing latency inevitably rise, since the available communication bandwidth and edge computing resources allocated to each user become more constrained. Furthermore, the latency reduction achieved by the proposed method, relative to several benchmark schemes, becomes increasingly significant with a larger number of AR users. This is attributed to the fact that the communication requirements grow more stringent as the user population increases, wherein the optimal configuration of the STAR-RIS plays a pivotal role in alleviating uplink communication latency.

\begin{figure}[htbp]
\begin{center}
\includegraphics[width=0.5\textwidth]{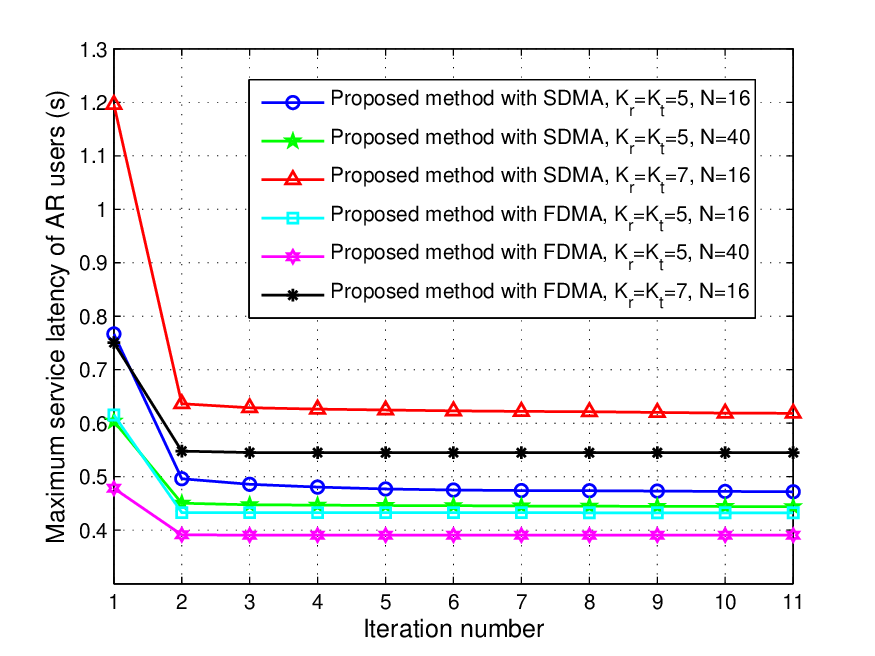}
\caption{Convergence of proposed algorithms.}
\end{center}
\end{figure}
In Fig. 8, we show the convergence of proposed alternating optimization methods, where $B=50$ MHz, $p_{k,\text{max}}=1$ W,  $d_k=150$ pixels, and $F=20$ GHz. As shown in this figure, the proposed alternating optimization methods tend to converge to a stationary point after several iterations under different simulation parameters. It further demonstrates the fast convergence speed of proposed methods.
\section{Conclusion}
This paper investigated the resource scheduling framework for STAR-RIS-assisted AR-enhanced Metaverse systems. The optimization framework can achieve the joint scheduling of the transmit power and CPU-cycle frequency of AR users, transmit/reflect coefficient matrix of the STAR-RIS, and computation resource allocation of Metaverse server located at the BS, for minimizing the service latency. Moreover, we developed an alternating optimization algorithm to solve the non-convex resource scheduling problem via rigorous derivation and analysis. Besides, the convergence and computationally complexity of proposed method were also presented. Numerical results revealed that our proposed STAR-RIS configuration and resource scheduling algorithm outperforms the existing benchmark methods in terms of maximum service latency among all AR users, especially in the scenario with restricted communication resources and heavy communication requirements. Moreover, our results indicate that SDMA serves as a suitable access scheme for STAR-RIS-assisted AR-empowered Metaverse systems under power-constrained or bandwidth-sufficient conditions, whereas FDMA demonstrates superior performance when the system bandwidth is restricted or when AR users are equipped with relatively higher transmission power.

Our work can be extended to several interesting research directions. Firstly, the general scenario with multiple BSs and STAR-RISs should be further investigated, with the joint optimization of user association control and resource management. Secondly, it is essential to study the robust optimization strategy to consider the imperfect CSI, since the presence of passive and massive STAR-RIS elements makes the channel estimation very challenging. Finally, this paper considers static users for analytical tractability, while future work will extend the proposed framework to incorporate user mobility, which is crucial for practical AR/Metaverse applications.
\bibliography{main}
\end{document}